\documentclass[11pt,a4paper]{article}

\usepackage[latin1]{inputenc}
\usepackage[ngerman,english]{babel}
\usepackage{graphicx}

\usepackage{amsmath,amssymb}
\usepackage{
						amsthm
           ,hyperref, bbm
           ,enumerate
           }
\usepackage[numbers,authoryear]{natbib}
\usepackage[usenames,dvipsnames]{color}
\usepackage[left=3cm,top=3.5cm,right=2cm,bottom=3cm]{geometry}
\newtheorem{theorem}{Theorem}

\newtheorem{corollary}[theorem]{Corollary}

\newtheorem{example}[theorem]{Example}

\newtheorem{proposition}[theorem]{Proposition}
\newtheorem{remark}[theorem]{Remark}

\definecolor{darkgreen}{rgb}{0.00,0.40,0.00}

\newcommand{\argmax}{\mathrm{argmax}}

\newcommand{\cb}[1]{\color{blue}#1\color{black}}
\renewcommand{\cb}[1]{#1}

\newcommand{\cbx}[1]{\color{blue}#1\color{black}}
\renewcommand{\cbx}[1]{#1}

\newcommand{\1}{\mathbbm{1}}

\newcommand{\IE}{\mathbb{E}}
\newcommand{\E}{\mathbb{E}}

\newcommand{\IN}{\mathbb{N}}

\newcommand{\cE}{\mathcal{E}}
\newcommand{\cF}{\mathcal{F}}

\newcommand{\std}{\mathsf{stddev}}

\title{Dual representations for general multiple stopping problems}
\author{	
	\small Christian Bender\\
        \footnotesize  Universit\"at des Saarlandes\\
        \footnotesize  Fachrichtung Mathematik\\
        \footnotesize  Postfach 151150\\
        \footnotesize  66041 Saarbr\"ucken \\
        \footnotesize  bender@math.uni-sb.de
        	\and
	\small John Schoenmakers\\
        \footnotesize  Weierstrass Institute for\\
        \footnotesize  Applied Analysis and Stochastics \\
        \footnotesize  Mohrenstr. 39\\
        \footnotesize  10117 Berlin \\
        \footnotesize  schoenma@wias-berlin.de
        	\and
	\small Jianing Zhang \\
        \footnotesize  Weierstrass Institute for\\
        \footnotesize  Applied Analysis and Stochastics \\
        \footnotesize  Mohrenstr. 39\\
        \footnotesize  10117 Berlin \\
        \footnotesize  zhang@wias-berlin.de
\vspace*{0.4cm}
}

\begin{document}

\def\linenumberfont{\normalfont\small\sffamily}
\selectlanguage{english}
\maketitle

\begin{abstract}
\noindent In this paper, we study the dual representation for generalized multiple stopping problems, hence the
pricing problem of \cb{general }multiple exercise options. We derive a dual representation which allows for cashflows which
are subject to volume constraints modeled by integer valued adapted processes and refraction periods modeled by
stopping times. As such, this extends the works by \cite{Schoen2010}, \cite{Bender2010}, \cite{Bender_vol},
\cite{AH2010}, and \cite{MH2004} on multiple exercise options, \cb{which }either take into consideration a refraction \cb{period }or
volume constraints, but not both simultaneously.
\cb{We also allow more flexible cashflow structures than the additive structure in the above references. For example some exponential
utility problems are covered by our setting. }We supplement the theoretical results with an explicit Monte Carlo algorithm
for constructing confidence intervals for the price of multiple exercise options and exemplify it by a numerical study
on the pricing of a swing option in \cb{an }electricity market.
\end{abstract}

\medskip


\medskip

%%%%%%%%%%%%%%%%%%%%%%%%%%%%%%%%%%%%
%%%%%%% section %%%%%%%%%%%%%%%%%%%%
%%%%%%%%%%%%%%%%%%%%%%%%%%%%%%%%%%%%
\section{Introduction}

The last decades have seen ground breaking developments of Monte Carlo methods for American options based on multidimensional underlying price processes.
In the late nineties the regression based methods by \cite{Car1996}, \cite{LS2001}, and \cite{TV2001} may be
considered as main breakthroughs.  In general these methods provide lower bounds \cb{on the option price }by constructing an approximation
to the optimal exercise (stopping) time via regression on a set of basis functions.  As such these approaches are termed ``primal''.
At the beginning of this century \cite{Rogers2002} and independently \cite{HK2004} provided the next
breakthrough by presenting a ``dual'' representation for the optimal stopping problem corresponding to
the American option pricing problem. In this representation the option price is expressed as infimum of an expectation
over a set of martingales. (The key behind this dual representation \cb{can }already be found in \cite{DK1994}, in fact.)
  While in the primal methods the central problem is to find a ``good'' stopping time, in
 the dual problem one needs to find a ``good'' martingale,  which \cb{leads to }an upper bound for the
price of an American option. As one of the standard numerical approaches to compute dual upper bounds
for American options by Monte Carlo we refer to \cite{AB2004}.

During the same time, in the emerging electricity markets products with a multiple of exercise opportunities, such as
 ``swing options'', became popular. Naturally, pricing of such a product leads to a multiple stopping problem,
and so numerical methods for solving multiple stopping problems were called for. In this respect,  generalization
of the existing primal regression methods for standard optimal stopping was just a matter of \cb{routine. }
Further, \cite{BenSch2006} developed a kind of policy iteration for  multiple stopping.
However, regarding the dual approach
the situation was not so clear. \cite{MH2004} proposed a dual representation for the multiple stopping problem via expressing the excess value due to each additional exercise right by an
infimum of an expectation over a set of martingales \cb{and }a set of stopping times.
This line of research was carried out further by \cite{AH2010} and \citet{Bender_vol} in the context of dual
 pricing of multi-exercise options under volume constraints.
%But, the question of a more natural dual representation in terms of martingales only was still open by then. The
%course was changed in \cite{Schoen2010} (2010) by a pure martingale dual representation,
\cb{Recently, \cite{Schoen2010} introduced a dual representation for the price of a multiple exercise option in contrast
to the dual representation for the excess value of an additional right. This new dual representation involves an infimum
over martingales only and can thus be considered as a more natural extension of the dual representation for single exercise options. This
approach  was generalized by \cite{Bender2010} to a continuous time setting involving (constant) refraction periods.}

In the meantime \cite{Quenez2010} introduced and studied multiple stopping problems \cb{in the primal sense }in a far more general context, where the payoff is
considered to be some abstract functional of an (ordered) sequence of stopping times.
The goal of \cb{the present paper }is to find (pure) martingale dual representations for \cb{such }generalized multiple
 stopping problems  in a discrete time setting. As we will show, such representations
can be constructed even in a most general setting. However, for practical \cb{implementation these general representations
unfold their full strength only, if applied to some more specifically structured cashflows. }In this respect we study a generic payoff structure with both multiplicative and additive structure that incorporates (integer valued)
volume constraints and refraction periods given by stopping \cb{times.  We }furthermore provide an explicit Monte Carlo
based algorithm and give a detailed numerical study exemplifying the pricing of swing options.
 Comparing to existing works, the numerical experiments reveal that, \cb{by and large, }the dual algorithms due to our new
 representations applied to the problem \cb{type }considered in \cite{AH2010} and \cite{Bender_vol} produce tighter upper bounds \cb{on
the option price, in particular when the number of exercise rights is large. }We moreover present a numerical example which involves swing options subject to both volume constraints and refraction periods and give tight confidence intervals for the respective option prices.
\cb{We underline }that the latter example cannot be treated by the \cb{dual }methods presented in the literature so far.

The structure of the paper is as follows: \cbx{In Section 2 we derive a dual representation for general multiple stopping problems
in terms of a family of martingales. As this family of martingales is typically too large for practical purposes in general,
we specialize to a generic
cashflow with additive and multiplicative structure which incorporates volume constraints and refraction periods in Section 3.
In this section we then prove two dual representations. One is in terms of the Doob decomposition of the Snell envelopes of
an auxiliary family of stopping
problems, the other one only requires approximations of these Snell envelopes. In Section 4 we explain how to build a Monte Carlo algorithm
for computing confidence intervals on the value of the multiple stopping problems based on the results of Section 3 and perform
some numerical experiments in the context of swing option pricing. }

%%%%%%%%%%%%%%%%%%%%%%%%%%%%%%%%%%%%
%%%%%%% section %%%%%%%%%%%%%%%%%%%%
%%%%%%%%%%%%%%%%%%%%%%%%%%%%%%%%%%%%
\section{General multiple stopping problem}

\cb{In this section we consider a multiple stopping problem in discrete time $i=0,\ldots, T$,
where $T\in \mathbb{N}$ is a fixed and finite time horizon. }We further introduce a \cb{``cemetery time'' $\partial:=T+1$ }where all rights will be exercised, which are not exercised up to time $T.$
For a given filtration $(\mathcal{F}_{i})_{0\leq i\leq\partial}$ and a number
$L$ of exercise dates we next consider a cashflow $X$ as a map
$X:\{0,...,T,\partial\}^{L}\times\Omega\rightarrow\mathbb{R}$ which satisfies \cb{for all $0\leq i_{1}\leq\cdot\cdot\cdot\leq i_{L}\leq\partial,$
\begin{gather*}
X_{i_{1},...,i_{L}} \textnormal{ is } \mathcal{F}_{i_{L}}\textnormal{-measurable},\\
%X\geq K\text{ \ for some \ }K\in\text{ }\mathbb{R},\text{ \ \ and}\\
\IE\left\vert
X_{i_{1},...,i_{L}}\right\vert <\infty.
\end{gather*}
}Now consider
the stopping problem (sup:=ess.sup, $\IE_{i}$ $:=$ $\IE_{{\cal F}_i}$),%
\[
Y_{i}^{\ast L}=\sup_{i\leq\tau^{1}\leq\cdot\cdot\cdot\leq\tau^{L}}%
\IE_{i}\,X_{\tau^{1},...,\tau^{L}},
\]
where the supremum runs over a family of ordered stopping times $\tau^{k}, {1 \leq k \leq L}$.

Let us define \cb{for $k=2,\ldots, L$ and $0\leq j_1\ldots\leq j_{k-1}\leq r\leq \partial$},
\begin{align}
Y_{r}^{\ast L-k+1,j_{1},\ldots,j_{k-1}}  :=\sup_{r\leq\tau^{k}\leq
\cdots\leq\tau^{L}}\IE_{r}X_{j_{1},\ldots,j_{k-1},\tau^{k},\ldots,\tau
^{L}}, \label{Reddef}
\end{align}
with the convention that for $k=1$, we put $Y_{r}^{\ast L,\varnothing}$ $:=$ $Y_{r}^{\ast L}$, \cb{and for $k=L+1$, we put
$Y_{r}^{\ast 0,j_{1},\ldots,j_{L}}=X_{j_1,\ldots, j_L}$.}

\medskip
\begin{proposition}
We have the following reduction principle
\begin{align}
\cb{Y_{r}^{\ast L-k+1,j_{1},\ldots,j_{k-1}} &
%=\sup_{r\leq\tau}\IE_{r}\underset{\tau\leq\tau^{(k+1)}\leq\cdots\leq
%\tau^{(L)}}{\,\sup_{\tau^{(k+1)},\cdots,\tau^{(L)}\text{ with}}}\IE_{\tau
%}X_{j_{1},\ldots,j_{k-1},\tau,\tau^{(k+1)}\ldots,\tau^{(L)}}\nonumber\\ &
=\sup_{\tau \geq r}\IE_{r}Y_{\tau}^{\ast L-k,j_{1},\ldots,j_{k-1},\tau},\text{
\ \ }r\geq j_{k-1}. \label{rp1}}
\end{align}
\end{proposition}
\begin{proof}
This principle can be straightforwardly proved in an inductive manner, but it can also be considered as a discrete time version of a related result in a continuous time setting from \cite{Quenez2010}.
\end{proof}

\medskip

In what follows the following remark turns out to be useful.
%%%%%%%%%%%%%%%%%%%%%%%%%%%%%%%%%%%%
%%%%%%% remark %%%%%%%%%%%%%%%%%%%%
%%%%%%%%%%%%%%%%%%%%%%%%%%%%%%%%%%%%
\begin{remark}
\label{Doob} {We say that a martingale $(M_{r})_{r\ge p}$ is a Doob martingale
of $(Y_{r})_{r\ge p},$ whenever there exists a predictable process
$(A_{r})_{r\ge p},$ such that $Y_{r}-M_{r}+A_{r}$ is $\mathcal{F}_{p}%
$-measurable for any $r\ge p$. In particular, for any two Doob martingales
$(M_{r})_{r\ge p}$ and $(\widetilde M_{r})_{r\ge p}$ of $(Y_{r})_{r\ge p},$ it
holds $$M_{r}-M_{r^{\prime}}=\widetilde M_{r}-\widetilde M_{r^{\prime}} = \sum_{k=r'}^{r-1} \big( Y_{\cb{k+1}} - \IE_k Y_{k+1} \big) $$ for
any $r \geq r^{\prime}\ge p.$ }
\end{remark}

\medskip

\cb{We can now state and prove a dual representation for the general multiple stopping problem in terms of martingales.}

%%%%%%%%%%%%%%%%%%%%%%%%%%%%%%%%%%%%
%%%%%%% theorem %%%%%%%%%%%%%%%%%%%%
%%%%%%%%%%%%%%%%%%%%%%%%%%%%%%%%%%%%
\begin{theorem}[\textbf{Dual representation}]
\label{dual}

In the setting described above, we have that:

\medskip

\textbf{(i)} \cb{For any $0\leq i \leq \partial$ and }any set of martingales $\left(  M_{r}^{L-k+1,j_{1}%
,...,j_{k-1}}\right)  _{r\geq j_{k-1}},$ where $1\leq k\leq L,$ and $\cb{i=:j_0}\leq
j_{1}\leq\cdot\cdot\cdot\leq j_{k-1}$, it holds
\begin{equation}
Y_{i}^{\ast L}\leq \IE_{i}\max_{i\leq j_{1}\leq\cdot\cdot\cdot\leq j_{L}%
\leq\partial}\left(  X_{j_{1},...,j_{L}}+\sum_{k=1}^{L}\left(  M_{j_{k-1}%
}^{L-k+1,j_{1},...,j_{k-1}}-M_{j_{k}}^{L-k+1,j_{1},...,j_{k-1}}\right)
\right). \label{eq:tmp001}
\end{equation}

\textbf{(ii)} It holds for $i\geq0$
\[
Y_{i}^{\ast L}=\max_{i\leq j_{1}\leq\cdot\cdot\cdot\leq j_{L}\leq\partial
}\left(  X_{j_{1},\ldots,j_{L}}+\sum_{k=1}^{L}\left(  M_{j_{k-1}}^{\ast
L-k+1,j_{1},\ldots,j_{k-1}}-M_{j_{k}}^{\ast L-k+1,j_{1},\ldots,j_{k-1}%
}\right)  \right)
\]
where for $1\leq k\leq L,$ and $\cb{i=:j_0}\leq j_{1}\leq\cdots\leq j_{k-1},$
$\left(  M_{r}^{\ast L-k+1,j_{1},\ldots,j_{k-1}}\right)  _{r\geq j_{k-1}}$ is
a Doob martingale of $\left(  Y_{r}^{\ast L-k+1,j_{1},\ldots,j_{k-1}}\right)
_{r\geq j_{k-1}}.$
\end{theorem}
%%%%%%%%%%%%%%%%%%%%%%%%%%%%%%%%%%
%%%%%%% proof %%%%%%%%%%%%%%%%%%%%
%%%%%%%%%%%%%%%%%%%%%%%%%%%%%%%%%%
\begin{proof}
{\textit{\textbf{(i)}}} For the martingale family as stated we have for any chain of stopping
times $0\leq\tau^{1}\leq\cdots\leq\tau^{L}\leq\partial,$%
\begin{align*}
&  \IE_{i}\,\sum_{k=1}^{L}\left(  M_{\tau^{k-1}}^{L-k+1,\tau^{1}%
,...,\tau^{k-1}}-M_{\tau^{k}}^{L-k+1,\tau^{1},\ldots,\tau^{k-1}%
}\right) \\
&  =\sum_{k=1}^{L}\IE_{i}\IE_{\tau^{k-1}}\left(  M_{\tau^{k-1}}^{L-k+1,\tau
^{1},\ldots,\tau^{k-1}}-M_{\tau^{k}}^{L-k+1,\tau^{1},...,\tau^{k-1}%
}\right)  =0,
\end{align*}
hence,
\begin{align*}
Y_{i}^{\ast L}= \sup_{i \leq \tau^{1} \leq \cdots \leq \tau^{L}} \IE_{i}\left(  X_{\tau^{1},...,\tau^{L}}+\sum_{k=1}^{L}\left(  M_{\tau^{k-1}}^{L-k+1,\tau
^{1},\ldots,\tau^{k-1}}-M_{\tau^{k}}^{L-k+1,\tau^{1},...,\tau^{k-1}%
}\right)\right)
\end{align*}
from which {(i)} follows directly.

\bigskip

\textit{\textbf{(ii)}} For any chain $i\leq j_{1}\leq\cdots\leq j_{L}\leq\partial$ we may
write (\cb{recalling }$j_{0}:=i$)
\begin{align}\label{rp}
&X_{j_{1},\ldots,j_{L}} +   \sum_{k=1}^{L}\left(  M_{j_{k-1}}^{\ast L-k+1,j_{1},\ldots,j_{k-1}} - M_{j_{k}}^{\ast L-k+1,j_{1},\ldots,j_{k-1}} \right) \nonumber\\
&  = X_{j_{1},\ldots,j_{L}} + \sum_{k=1}^{L}\left(  Y_{j_{k-1}}^{\ast L-k+1,j_{1},\ldots,j_{k-1}} - Y_{j_{k}}^{\ast L-k+1,j_{1},\ldots,j_{k-1}
}\right) \nonumber\\
&  +\sum_{k=1}^{L}\sum_{l=j_{k-1}}^{j_{k}-1}\left(  \IE_{l}Y_{l+1}^{\ast L-k+1,j_{1},\ldots,j_{k-1}}-Y_{l}^{\ast L-k+1,j_{1},\ldots,j_{k-1}}\right) \nonumber\\
&  =  Y_{i}^{\ast L} +  \sum_{k=1}%
^{\cb{L}}\left(  Y_{j_{k}}^{\ast L-k,j_{1},\ldots,j_{k}} - Y_{j_{k}}^{\ast L-k+1,j_{1},\ldots,j_{k-1}}\right) \nonumber\\
&  +\sum_{k=1}^{L}\sum_{l=j_{k-1}}^{j_{k}-1}\left(  \IE_{l}Y_{l+1}^{\ast L-k+1,j_{1},\ldots,j_{k-1}} - Y_{l}^{\ast L-k+1,j_{1},\ldots,j_{k-1}}\right) .
\end{align}
By the reduction principle (\ref{rp1}) it follows that \cb{$\left(  Y_{r}^{\ast L-k+1,j_{1},\ldots,j_{k-1}%
}\right)  _{r\geq j_{k-1}}$ is a supermartingale which dominates the (virtual)
cash-flow $Y_{r}^{\ast L-k,j_{1},\ldots,j_{k-1},r}$ for $k=1,...,L.$ Hence, }expression (\ref{rp}) is less than or equal to $Y_{i}^{\ast L}.$
 It thus follows that
\[
\max_{i\leq j_{1}\leq\cdot\cdot\cdot\leq j_{L}\leq\partial}\left(
X_{j_{1},\ldots,j_{L}}+\sum_{k=1}^{L}\left(  M_{j_{k-1}}^{\ast L-k+1,j_{1}%
,\ldots,j_{k-1}}-M_{j_{k}}^{\ast L-k+1,j_{1},\ldots,j_{k-1}}\right)  \right)
\leq Y_{i}^{\ast L},
\]
and then, an application of (i) finishes the proof.
\end{proof}
%%%%%%%%%%%%%%%%%%%%%%%%%%%%%%%%%%%%%%
%%%%%%% end proof %%%%%%%%%%%%%%%%%%%%
%%%%%%%%%%%%%%%%%%%%%%%%%%%%%%%%%%%%%%

A straightforward consequence of Theorem \ref{dual} is the following \cb{dual representation in terms of approximate Snell
envelopes.}

%%%%%%%%%%%%%%%%%%%%%%%%%%%%%%%%%%
%%%%%%% corollary %%%%%%%%%%%%%%%%
%%%%%%%%%%%%%%%%%%%%%%%%%%%%%%%%%%
\begin{corollary}\label{cor:tmp_001}
For any set $\left(Y_{r}^{L-k+1,j_{1},...,j_{k-1}}\right)_{r\geq j_{k-1}}$ of approximations to
the Snell envelopes $\big(  Y_{r}^{\ast L-k+1,j_{1},\ldots,j_{k-1}}\big)_{r\geq j_{k-1}}$ with $Y_{j_{L}}^{0,j_{1},\ldots,j_{L}}:=X_{j_{1}%
,\ldots,j_{L}},$ it holds for $i\geq0$
\begin{align}
Y_{i}^{\ast L}  &  \leq Y_{i}^{L}+\IE_{i}\max_{i\leq j_{1}\leq\cdot\cdot
\cdot\leq j_{L}\leq\partial}\sum_{k=1}^{L}\Bigg(  Y_{j_{k}}^{L-k,j_{1}%
,\ldots,j_{k}}-Y_{j_{k}}^{L-k+1,j_{1},\ldots,j_{k-1}} \nonumber\\
&  +\sum_{l=j_{k-1}}^{j_{k}-1}\left(  \IE_{l}Y_{l+1}^{L-k+1,j_{1}%
,\ldots,j_{k-1}}-Y_{l}^{L-k+1,j_{1},\ldots,j_{k-1}}\right)  \Bigg).
\label{alt}%
\end{align}
Equality holds when the Snell envelopes are plugged in.
\end{corollary}
\begin{proof}
%This result is obtained by plugging the Doob martingale of $\left(Y_{r}^{L-k+1,j_{1},...,j_{k-1}}\right)_{r\geq j_{k-1}}$ into the righthand side of \eqref{eq:tmp001} and following the same manipulations as in part \textbf{\textit{(ii)}} of the proof of Theorem \ref{dual}.

Given $\left(Y_{r}^{L-k+1,j_{1},...,j_{k-1}}\right)_{r\geq j_{k-1}}$, we \cb{denote a }corresponding family of Doob
martingales by $\left(M_{r}^{L-k+1,j_{1},...,j_{k-1}}\right)_{r\geq j_{k-1}}$. \cb{Following the same manipulations as in (\ref{rp})
and recalling that by definition $Y_{j_{L}}^{0,j_{1},\ldots,j_{L}}=X_{j_{1}
,\ldots,j_{L}}$, we get }
\begin{eqnarray*}
&&Y_{i}^{L} +  \sum_{k=1}%
^{{L}}\Biggl(\left(  Y_{j_{k}}^{ L-k,j_{1},\ldots,j_{k}} - Y_{j_{k}}^{ L-k+1,j_{1},\ldots,j_{k-1}}\right)
  \\ && \quad +\sum_{l=j_{k-1}}^{j_{k}-1}\left(  \IE_{l}Y_{l+1}^{L-k+1,j_{1},\ldots,j_{k-1}} - Y_{l}^{L-k+1,j_{1},\ldots,j_{k-1}}\right) \Biggr)
\\ &=& X_{j_{1},\ldots,j_{L}} +   \sum_{k=1}^{L}\left(  M_{j_{k-1}}^{ L-k+1,j_{1},\ldots,j_{k-1}} - M_{j_{k}}^{L-k+1,j_{1},\ldots,j_{k-1}} \right)
\end{eqnarray*}
\cb{Hence, the assertion is a mere reformulation of Theorem~\ref{dual}}.

\end{proof}

\medskip

At this point, we stress that the dual representation from Theorem \ref{dual} relies on families
of martingales $\left(M_{r}^{L-k+1,j_{1},...,j_{k-1}}\right)_{r\geq j_{k-1}}$ \cb{whose size is }parametrized
via \cb{the }$(k-1)$-tuples $(j_1,\ldots,j_{k-1}), k=1,\ldots,L$. Hence, depending on
\cb{the time horizon $T$ and the number of exercise rights $L$ a huge  number of
martingales $M^{L-k+1,j_{1},...,j_{k-1}}$, $k=1,\ldots, L$, $0\leq j_1 \leq \ldots\leq j_{k-1}\leq \partial=T+1$
is required in order to compute an upper price bound by the above dual formulation. }
It is thus of great importance to \cb{single out situations, in which a family of optimal martingales can be constructed from
a much smaller family of auxiliary processes. } This will be the topic of Section~\ref{GE}.
\cb{A motivating example in this respect is the standard multiple stopping problem.}

%The following well-known example illustrates how the Doob martingale based dual representation from Theorem \ref{dual}-\textit{\textbf{(ii)}} can be reduced to an equivalent dual representation in terms of Snell envelopes.

\medskip

%%%%%%%%%%%%%%%%%%%%%%%%%%%%%%%%
%%%%%%% example %%%%%%%%%%%%%%%%
%%%%%%%%%%%%%%%%%%%%%%%%%%%%%%%%
\begin{example}[\textbf{Standard multiple stopping}]\label{ex:std} Let $Z$ {be a nonnegative adapted
process with }$Z_{j}=0$ for $j=\partial,$ \cb{i.e. no penalty is imposed for unexercised rights. The standard multiple stopping problem
is to maximize $\IE[\sum_{k=1}^L Z_{\tau^{k}}]$ over the set of ordered stopping times $\tau^{1}\leq \cdots \leq \tau^{L}$ such
that $\tau^{k}<\tau^{k+1}$ or
$\tau^{k}=\tau^{k+1}=\partial$. This means that at most one right can be exercised per day, but an arbitrary number of rights can be left
unexercised, i.e. is exercised at time $\partial$. This problem can be put into our general setting by
considering }the cashflow
\[
X_{i_{1},...,i_{L}}=\begin{cases}
				\sum_{k=1}^{L}Z_{i_{k}}, &\text{if } \cb{i_{j+1}=i_j  \Rightarrow i_j=\partial\, , }\\
											-N,												&\text{else\, ,}
										\end{cases}
\]
\cb{for $N\in \mathbb{N}$. Note that the Snell envelope $Y^{\ast L}_i$ does not depend on the choice of $N$, because it is
never optimal to exercise $X$ in a way which gives a negative payment. Hence, letting $N$ tend to $\infty$,
}Theorem~\ref{dual} yields,

\[
Y_{i}^{\ast L}=\max_{\substack{i\leq j_{1}\leq \cdots\leq j_{L} \\ j_k=j_{k+1}\,\Rightarrow\, j_k=\partial}}\left(  \sum_{k=1}^{L}Z_{j_{k}%
}+\sum_{k=1}^{L}\left(  M_{j_{k-1}}^{\ast L-k+1,j_{1},\ldots,j_{k-1}}%
-M_{j_{k}}^{\ast L-k+1,j_{1},\ldots,j_{k-1}}\right)  \right)
\]

where for $1\leq k\leq L,$ $i\leq j_{1}<\cdot\cdot\cdot<j_{k-1},$
$\left(  M_{r}^{\ast L-k+1,j_{1},\ldots,j_{k-1}}\right)_{r\geq j_{k-1}}$ is
a Doob martingale of
\begin{align*}
Y_{r}^{\ast L-k+1,j_{1},\ldots,j_{k-1}}
%&  :=\underset{r\leq\tau^{(k)}%
%<\cdots<\tau^{(L)},\text{ \ }j_{k-1}<\tau^{(k)}}{\,\sup_{\tau^{(k)}%
%,\cdots,\tau^{(L)}\text{ with}}}\IE_{r}\left(  \sum_{p=1}^{k-1}Z_{j_{p}}%
%+\sum_{p=k}^{L}Z_{\tau^{(p)}}\right) \\ &
=\sum_{p=1}^{k-1}Z_{j_{p}}+\underset{\tau^{p}=\tau^{p+1}\,\Rightarrow\, \tau^{p}=\partial\text{ and }\tau^{k}=j_{k-1}\,\Rightarrow\,j_{k-1}=\partial}{\,\sup_{r\leq \tau^{k}
\leq\cdots\leq\tau
^{L}}}\IE_{r}\sum_{p=k}^{L}Z_{\tau^{p}}.
\end{align*}
for $r\geq j_{k-1}.$
\cb{Define, for $r\geq 0$,
\begin{align*}
Y_{r}^{\ast L-k+1} := \underset{\tau^{p}=\tau^{p+1}\,\Rightarrow\, \tau^{p}=\partial}{\,\sup_{r\leq \tau^{k}%
\cdots\leq \tau^{L}}}\IE_{r}\left( \sum_{p=k}^{L}Z_{\tau^{p}}\right)
\end{align*}
and denote the Doob martingale of $(Y_{r}^{\ast L-k+1})_{r\geq 0}$ by $(M_{r}^{\ast L-k+1})_{r\geq 0}$.
As $$Y_{r}^{\ast L-k+1}-Y_{r}^{\ast L-k+1,j_{1},\ldots,j_{k-1}}$$ is $\mathcal{F}_{j_{k-1}}$-measurable for
$r\geq j_{k-1}$, we can conclude by Remark \ref{Doob} that $(M_{r}^{\ast L-k+1})_{r\geq j_{k-1}}$ is a Doob martingale
of $\left(  Y_{r}^{\ast L-k+1,j_{1},\ldots,j_{k-1}}\right)  _{r\geq j_{k-1}}$. Hence we end up with the dual representation
\[
Y_{i}^{\ast L}=\max_{\substack{i\leq j_{1}\leq \cdots\leq j_{L} \\ j_k=j_{k+1}\,\Rightarrow\, j_k=\partial}}\left(  \sum_{k=1}^{L}Z_{j_{k}%
}+\sum_{k=1}^{L}\left(  M_{j_{k-1}}^{\ast L-k+1}%
-M_{j_{k}}^{\ast L-k+1}\right)  \right)
\]
of \cite{Schoen2010}. Here the potentially large family of optimal  martingales $\left(  M^{\ast L-k+1,j_{1},\ldots,j_{k-1}}\right)$,
$k=1,\ldots, L$, $0\leq j_1,\ldots \leq j_{k-1}\leq \partial$, collapses, in fact, to a family of $L$ martingales, namely the Doob martingales
of $Y^{\ast k}$.}
\end{example}

\bigskip

%%%%%%%%%%%%%%%%%%%%%%%%%%%%
%%%%% section %%%%%%%%%%%%%%
%%%%%%%%%%%%%%%%%%%%%%%%%%%%
\section{Generic cashflow with additive and multiplicative structure} \label{GE}

\cb{We now introduce a generic cashflow structure for which the dual representation simplifies in a similar way than for the standard multiple stopping problem in Example \ref{ex:std}. To this end }let us consider for each $k=1,...,L$ and $l=1,...,L-1$ two adapted processes
$U^{k}$ and $V^{l}$. We define a ``pre-cashflow''
\[
\widetilde{X}_{j_{1},...,j_{L}}=\sum_{k=1}^{L}U_{j_{k}}^{k}
{\displaystyle\prod\limits_{l=1}^{k-1}}
V_{j_{l}}^{l},
\]
which \cb{is assumed to satisfy }$\widetilde{X}_{j_{1},...,j_{L}} > -N$ for some (possibly large) $N\in\IN$.
\cb{Concerning the processes $U^{k}$ and $V^{l}$, we suppose that $U^{k}_i$ is integrable for every $k=1,\ldots,L$ and $i=0,\ldots,\partial$, and that
$V^{l}_i$ is strictly positive and bounded from above for every $l=1,\ldots,L-1$ and $i=0,\ldots,\partial$. }
\cb{The multiple stopping problem which we have in mind is to optimally exercise this pre-cashflow under some constraints on the set of admissible stopping times, which we now formulate. We first }define an adapted volume constraint process $v$ with values in
$\{1,...,L\}$ such that $v_{t}$ is the maximum number of rights
one may exercise at $t,$ \cb{and such that $v_{\partial}=L$. In order to formalize this constraint, }we \cb{introduce }for $p\geq 1$ the mapping $\cE_p$ which acts on a non-decreasing $p$-tuple $(j_1,\ldots,j_p)$ by
\begin{align*}
\mathcal{E}_{p}(j_{1},...,j_{p})&:=\#\{r:1\leq r\leq p,\text{ }j_{r}=j_{p}\}.
%\mathcal{E}_{L-p+1}(j_{p},...,j_{L})&:=\#\{r:p\leq r\leq L,\text{ }j_{r}=j_{L}\},
\end{align*}

Hence, $\cE_p$ denotes the number of rights exercised at $j_{p}$ in the non-decreasing chain $0\leq j_{1}%
\leq\cdot\cdot\cdot\leq j_{p}\leq\partial.$
\cb{Obviously, an ordered chain of stopping times $\tau^{1}\leq \cdots \leq \tau^{L}$ satisfies the volume
 constraint if and only if $\mathcal{E}_{p}(\tau^{1},...,\tau^{p})\leq v_{\tau^{p}}$ for every $p=1,\ldots,L$.
The second constraint, which we want to impose, is a refraction period which specifies the minimal waiting time between
two exercises at different times. We admit random refraction periods, i.e. }at each time $i,$ $0\leq
i<\partial,$ \cb{we fix }a stopping time
$\rho^{i}$ taking values in $\{i+1,...,\partial\}.$  \cb{If at least one right is exercised at time $i$, then the refraction period constraint imposes that the next right must either be exercised at the same time (if consistent with the volume constraint) or otherwise no earlier than $\rho^{i}$. }A standard case is
$\rho^{i}=\left(  i+\delta\right)  \wedge\partial,$ where $1\leq \delta \leq T$ is
deterministic. \cb{Both constraints can be summarized by }the binary $\mathcal{F}_{j_{p}}$-measurable random variable%

\[
\mathcal{C}_{p}(j_{1},...,j_{p}):=\begin{cases}
																		1,	&\forall_{1\leq l\leq p}: ~ \mathcal{E}_{l}(j_{1},...,j_{l})\leq v_{j_l} \text{ and } \forall_{1\leq l\leq p}: ~ j_{l}>j_{l-1}\Longrightarrow j_{l}\geq\rho^{j_{l-1}}\\	
																		0, &\text{else,}
\end{cases}
\]
\cb{which is equal to 1, if and only if the constraints are satisfied when exercising at the $p$ times $j_1\leq \cdots\leq j_p$.}

\cb{The dynamic multiple stopping problem which we now study is
\begin{equation}\label{eq:generic}
Y^{*L}_i=\sup_{\substack{i\leq\tau^{1}\leq\cdots \leq \tau^{L}\leq \partial \\ \mathcal{C}_{L}(\tau^{1},...,\tau^{L})=1 }} \IE_i\left[ \sum_{k=1}^{L}U_{\tau^{k}}^{k}
{\displaystyle\prod\limits_{l=1}^{k-1}}
V_{\tau^{l}}^{l},\right]
\end{equation}
i.e. the supremum is taken over all stopping times with values in $\{i,\ldots,T,\partial\}$ which satisfy the volume constraint and the refraction period constraint. This problem fits in our general (unconstrained) setting by considering the cashflow}
%Later on, shall also need
%
%\[
%\mathcal{C}_{L-p+1}(j_{p},...,j_{L}):=\begin{cases}
%																		1,	&\forall_{p\leq l\leq L}: ~ \mathcal{E}_{l-p+1}(j_{p},...,j_{l})\leq v_{j_l} \text{ and } \forall_{p\leq l\leq L}: ~ j_{l}>j_{l-1}\Longrightarrow j_{l}\geq\rho^{(j_{l-1})}\\	
%																		0, &\text{else.}
%\end{cases}
%\]

\begin{align}\label{eq:tmp018}
X_{j_{1},...,j_{L}}=\begin{cases}
											\widetilde{X}_{j_{1},...,j_{L}},	&\text{if }\mathcal{C}_{L}(j_{1},,...,j_{L})=1,\\
											-N,	&\text{else.}
\end{cases}
\end{align}

To illustrate our motivation for studying the previous cashflow, let us have a look at the following examples.

\medskip

%%%%%%%%%%%%%%%%%%%%%%%%%%%
%%%%%% examples %%%%%%%%%%%
%%%%%%%%%%%%%%%%%%%%%%%%%%%

%%%%%%%%%%%%%%%%%%%%%%%%%%%%%%%%
\begin{example}[\textbf{Swing options}]
\cb{We extend the situation in Example \ref{ex:std} by imposing volume constraints and refraction periods as decribed above. Hence, we have}
\begin{align*}
&V^{l}_j := 1, ~ l=1,\ldots,L-1, ~ j=0,\ldots,\partial,\\
&U^{p}_j := Z_j  ~ p=1,\ldots,L, ~ j=0,\ldots,\partial,
\end{align*}
\cb{where we recall that $Z$ is a nonnegative adapted process with $Z_\partial=0$. The multiple stopping problem then becomes
$$
\sup_{\substack{\tau^{1}\leq\cdots \leq \tau^{L} \\ \mathcal{C}_{L}(\tau^{1},...,\tau^{L})=1 }} \IE\left[ \sum_{k=1}^{L} Z_{\tau^{k}}\right],
$$
leading to
\begin{align*}
X_{j_{1},...,j_{L}}=\begin{cases}
											\sum_{k=1}^{L} Z_{j_k},	&\text{if }\mathcal{C}_{L}(j_{1},,...,j_{L})=1,\\
											-N,	&\text{else.}
\end{cases}
\end{align*}
Here any $N\in \mathbb{N}$ can be chosen because $Z$ is  nonnegative.
A dual approach for this multiple stopping problem was studied by \citet{Bender_vol} and \citet{AH2010} under volume constraints,
 but with unit refraction period, i.e. $\rho^{i}=i+1$. The case with non-trivial constant refraction period
is treated in \citet{Bender2010}, but only under unit volume constraint, i.e. $v_i=1$. A typical problem in the
context of electricity markets which leads to this type of multiple stopping problem is the pricing of Swing option contracts, in which volume constraints and refraction periods are often imposed. This option pricing problem will be explained in more detail in our numerical study in Section \ref{sec:num_ex}.}
%Swing options with volume constraints and refraction periods can be embedded into our framework. A stylized version of a swing option entitles the buyer to buy on each day $j$ over a period of $T$ days up to $v_j$ packages of electricity for a fixed strike price. Once the rights are exercised within the volume constraints, the buyer has to wait a period of $\delta \geq 1$ days until electricity can be bought again. Swing options with volume constraints under unit refraction period, studied in \cite{MH2004} and \cite{Bender2010}, as well as swing options with constant refraction period under unit volume constraints, studied in \cite{Bender_vol}, follow as special instances from our generic cashflow framework. More precisely, assume that
%\begin{align*}
%&V^{(l)}_j := 1, ~ l=1,\ldots,L-1, ~ j=0,\ldots,T,\\
%&U^{(p)}_j := Z_j := Z(S_j):= (S_j - K)^+, ~ j=0,\ldots,T, ~ p=1,\ldots,L,
%\end{align*}
%i.e. $S_j$ is the price of electricity and $Z_j$ the payoff. This reproduces the swing option setting of the aforementioned papers. In Section \ref{sec:num_ex}, we present an extensive numerical study for the pricing of swing options and also give a comparison to existing results from the literature.
\end{example}
%%%%%%%%%%%%%%%%%%%%%%%%%%%%%%%%%

\medskip

\begin{example}[\textbf{\cb{Exponential utility}}]
\label{ex:exp_utility} \cb{Under the assumptions of the previous example we can also maximize the exponential utility of
 exercising the cashflow $Z_i$ $L$-times while obeying the constraints. Given the risk aversion parameter $\alpha \in (0,\infty)$ the corresponding multiple stopping problem becomes
$$
\sup_{\substack{\tau^{1}\leq\cdots \leq \tau^{L} \\ \mathcal{C}_{L}(\tau^{1},...,\tau^{L})=1 }} \IE\left[ - e^{-\alpha \sum_{k=1}^{L} Z_{\tau^{k}}}\right].
$$

This problem fits in our setting by considering
\begin{align*}
X_{j_{1},...,j_{L}}=\begin{cases}
											\sum_{k=1}^{L}U_{j_{k}}^{k}
{\displaystyle\prod\limits_{l=1}^{k-1}}
V_{j_{l}}^{l},,	&\text{if }\mathcal{C}_{L}(j_{1},,...,j_{L})=1,\\
											-N,	&\text{else}
\end{cases}
\end{align*}
with}
\begin{align*}
&V^{l}_j := e^{-\alpha Z_j}>0 \quad \text{ and } \quad U^{k}_j := \begin{cases}
							0,	&\text{if  } k=1,\ldots,L-1,\\
							\cb{-e^{-\alpha Z_j}},	&\text{if  }k=L,
							\end{cases}
\end{align*}
for $j=0,1,\dots,\partial$ \cb{and $N\geq 2$.}
\end{example}
%%%%%%%%%%%%%%%%%%%%%%%%%%%%%%
%%%%%%%%%%%%%%%%%%%%%%%%%%%%%%

\medskip
%%%%%%%%%%%%%%%%%%%%%%%%%%%%%%%%%
\begin{example}[\textbf{Portfolio liquidation}]
Suppose a (large) investor on a illiquid market wants to sell out
(liquidate) $L$ shares of a stock \cb{during the period }$\{0,...,T\}.$ \cb{We assume that $ \widetilde{S}_{j}>0  $, $j=0,\ldots,T$,  is the virtual stock price process
reflecting the stock price evolution in the absence of the large investor's trading. In the spirit of
\citet{SchSch}, Section 3.1, we model the price impact of the large investor by a resilience function $G$ which we here apply to the log-price. } Hence, the log-stock price $\ln S_{j_{k}}^{j_{1}%
,...,j_{k-1}}$ \cb{at  time $j_k$ of the sale }of the $k$th share, where $k-1$ shares were
already sold at dates $0\leq j_{1}\leq\cdot\cdot\cdot\leq j_{k-1},$ is given
by%
\[
\ln S_{j_{k}}^{j_{1},...,j_{k-1}}=\ln\widetilde{S}_{j_{k}}-\cb{\sum_{l=1}^{k-1}%
G(j_k-j_{l}).}
\]
\cb{We here choose the capped linear resilience function $G(t)=b(1-at)_+$ for constants $a,b >0$. Assuming a short time horizon $T\leq 1/a$ },
the
investor is thus faced with a multiple stopping problem
$$
\cb{\sup_{\tau^{1}\leq\cdots \leq \tau^{L} } \IE\left[\sum_{k=1}^L S_{\tau^{k}}^{\tau^{1},...,\tau^{k-1}} \right],}
$$
\cb{which fits in our framework by applying, for $0\leq j_1\leq\cdots\leq j_L\leq T$, the cashflow}%
\begin{align*}
X_{j_{1},...,j_{L}}  & :=\sum_{k=1}^{L}S_{j_{k}}^{j_{1},...,j_{k-1}}%
=\sum_{k=1}^{L}\widetilde{S}_{j_{k}}\exp\left(  -\sum_{l=1}^{k-1}b\left(
1-a\left(  j_{k}-j_{l}\right)  \right)  \right)  \\
& =\sum_{k=1}^{L}\widetilde{S}_{j_{k}}\exp\left[  b\left(  a\,j_{k}-1\right)
(k-1)\right]
%TCIMACRO{\dprod \limits_{l=1}^{k-1}}%
%BeginExpansion
{\displaystyle\prod\limits_{l=1}^{k-1}}
%EndExpansion
\exp\left(  -abj_{l}\right)  \\
& =\sum_{k=1}^{L}U_{j_{k}}^{k}%
%TCIMACRO{\dprod \limits_{l=1}^{k-1}}%
%BeginExpansion
{\displaystyle\prod\limits_{l=1}^{k-1}}
%EndExpansion
V_{j_{l}}^{l}%
\end{align*}
with%
\[
U_{j}^{k}:=\widetilde{S}_{j}\exp\left[  b\left(  aj-1\right)
(k-1)\right]  \text{ \ \ and \ \ }V_{j}^{l}:=\exp\left(  -abj\right)  .
\]
\cb{(Note that the cemetery time $\partial$ is irrelevant in this setting and we can e.g. set $U^{k}_\partial=0$, $V^{l}_\partial=1$ to make sure that it is never optimal to exercise at this time).}
\end{example}
%%%%%%%%%%%%%%%%%%%%%%%%%%%%%%%%%%

\bigskip

\cb{Similarly to the situation in Example \ref{ex:std}, we now introduce a family of auxiliary multiple stopping problems $Y_{r}^{\ast L-k+1}$, which are not
parameterized by the times $j_1,\ldots, j_{k-1}$, at which the first rights were exercised. We will then show that a family of optimal martingales
for the original multiple stopping problem (\ref{eq:generic}) can be constructed via the Doob decomposition  of  the auxiliary problems.
This then leads to a simplified dual representation for (\ref{eq:generic}), which can be implemented in practice even when the maturity $T$ and the number
of rights $L$ are large.

Define}

\begin{equation}
Y_{r}^{\ast L-k+1}:=\underset{r\leq\tau^{k}\leq\cdots\leq\tau^{L}%
\text{\ and }\mathcal{C}_{L-k+1}(\tau^{k},...,\tau^{L})=1}{\,\sup
_{\tau^{k},\cdots,\tau^{L}}}\IE_{r}\left(  \sum_{p=k}^{L}%
U_{\tau^{p}}^{p}%
%TCIMACRO{\dprod \limits_{l=k}^{p-1}}%
%BeginExpansion
{\displaystyle\prod\limits_{l=k}^{p-1}}
%EndExpansion
V_{\tau^{l}}^{l}\right) \label{df}%
\end{equation}
with the convention $Y_{r}^{\ast0}:=0.$
\cb{The following proposition states the Bellman principle for this multiple stopping problem.}

\begin{proposition} [\bf Dynamic program]
\label{prop:dynamic_prog}
For $r\ge 0$ and $1\le k\le L$ we have,

\[
Y_{r}^{\ast L-k+1}=\max\left(\IE_{r}Y_{r+1}^{\ast L-k+1},%
%TCIMACRO{\dbigvee }%
%BeginExpansion
%{\displaystyle\bigvee}
%EndExpansion
\max_{1\leq n\leq v_{r}\,\wedge\,\,(L-k+1)}\left(  \sum_{p=k}%
^{k+n-1}U_{r}^{p}%
%TCIMACRO{\dprod \limits_{l=k}^{p-1}}%
%BeginExpansion
{\displaystyle\prod\limits_{l=k}^{p-1}}
%EndExpansion
V_{r}^{l}+%
%TCIMACRO{\dprod \limits_{l=k}^{k+n-1}}%
%BeginExpansion
{\displaystyle\prod\limits_{l=k}^{k+n-1}}
%EndExpansion
V_{r}^{l}\IE_{r}Y_{\rho^{r}}^{\ast L-k-n+1}\right)\right)  .
\]
\end{proposition}

\begin{proof}
From \eqref{df} we derive straightforwardly,
\cb{
\begin{eqnarray*}
&&Y_{r}^{\ast L-k+1}  \\
& =& \max_{0\leq n\leq v_{r}\,\wedge\,(L-k+1)}
\sup_{\substack{r< \tau^{k+n}\leq\cdots \leq \tau^{L} \\ \mathcal{C}_{L-k+1}(r,\ldots,r,\tau^{k+n},\ldots,\tau^{L})=1 }}
 \IE_{r}\left(  \sum_{p=k}^{k+n-1}U_{r}^{p}%
%TCIMACRO{\dprod \limits_{l=k}^{p-1}}%
%BeginExpansion
{\displaystyle\prod\limits_{l=k}^{p-1}}
%EndExpansion
V_{r}^{l}+\sum_{p=k+n}^{L}U_{\tau^{p}}^{p}%
%TCIMACRO{\dprod \limits_{l=k}^{p-1}}%
%BeginExpansion
{\displaystyle\prod_{l=k}^{k+n-1} V_{r}^{l}\prod\limits_{l=k+n}^{p-1}}
%EndExpansion
V_{\tau^{l}}^{l}\right)  \\
&  =&\max\Biggl( \IE_{r}Y_{r+1}^{\ast L-k+1},
 \max_{1\leq n\leq v_{r}\,\wedge\,(L-k+1)} \Biggl(\sum_{p=k}^{k+n-1}U_{r}^{p} {\displaystyle\prod\limits_{l=k}^{p-1}}
V_{r}^{l} \\&& +\prod_{l=k}^{k+n-1} V_{r}^{l}
\sup_{\substack{\rho^{r}\leq \tau^{k+n}\leq\cdots \leq \tau^{L} \\ \mathcal{C}_{L-k-n+1}(\tau^{k+n},\ldots,\tau^{L})=1 }}
  \IE_{r}\Bigl(\sum_{p=k+n}^{L}U_{\tau^{p}}^{p}%
{\displaystyle\prod\limits_{l=k+n}^{p-1}}
V_{\tau^{l}}^{l}\Bigr) \Biggr)\Biggr).
\end{eqnarray*}
A standard argument shows that
\begin{eqnarray}\label{eq:hilf1}
&& \sup_{\substack{\rho^{r}\leq \tau^{k+n}\leq\cdots \leq \tau^{L} \\ \mathcal{C}_{L-k-n+1}(\tau^{k+n},\ldots,\tau^{L})=1 }}
  \IE_{r}\Bigl(\sum_{p=k+n}^{L}U_{\tau^{p}}^{p}%
{\displaystyle\prod\limits_{l=k+n}^{p-1}}
V_{\tau^{l}}^{l}\Bigr) \nonumber \\ &=&\IE_{r} \sup_{\substack{\rho^{r}\leq \tau^{k+n}\leq\cdots \leq \tau^{L} \\ \mathcal{C}_{L-k-n+1}(\tau^{k+n},\ldots,\tau^{L})=1 }}
  \IE_{\rho^{r}}\Bigl(\sum_{p=k+n}^{L}U_{\tau^{p}}^{p}%
{\displaystyle\prod\limits_{l=k+n}^{p-1}}
V_{\tau^{l}}^{l}\Bigr)=\IE_{r}Y_{\rho^{r}}^{\ast L-k-n+1},
\end{eqnarray}}
\cb{which concludes the proof.}
\end{proof}

We now establish \cb{a }crucial relationship between the Snell envelopes $Y^{\ast
L-k+1,j_{1},\ldots,j_{k-1}}$ and $Y^{\ast L-k+1}$ defined in (\ref{df}). The following Proposition shows that $Y^{\ast
L-k+1,j_{1},\ldots,j_{k-1}}$, parameterized by the $j_k$'s, can be represented in terms of $Y^{\ast L-k+1}$ which avoids the $j_k$'s.
\cb{Notice that, for $k=1$, both Snell envelopes coincide by definition.}
\begin{proposition}\label{Prop10}
\cb{Suppose $1< k\le L+1$. }Under the condition $\mathcal{C}_{k-1}(j_{1},...,j_{k-1})=1$,
we have
\begin{itemize}
\item[(i)]
for $r>j_{k-1}$ it holds
\begin{equation}
Y_{r}^{\ast L-k+1,j_{1},\ldots,j_{k-1}}=\sum_{p=1}^{k-1}U_{j_{p}}^{p}%
%TCIMACRO{\dprod \limits_{l=1}^{p-1}}%
%BeginExpansion
{\displaystyle\prod\limits_{l=1}^{p-1}}
%EndExpansion
V_{j_{l}}^{l}+\IE_{r}Y_{\rho^{j_{k-1}}\vee  r}^{\ast L-k+1}%
%TCIMACRO{\dprod \limits_{l=1}^{k-1}}%
%BeginExpansion
{\displaystyle\prod\limits_{l=1}^{k-1}}
%EndExpansion
V_{j_{l}}^{l}. \label{xp1}%
\end{equation}

\item[(ii)] Further it holds
\begin{eqnarray}
&&Y_{j_{k-1}}^{\ast L-k+1,j_{1},\ldots,j_{k-1}}  =\sum_{p=1}^{k-1}U_{j_{p}%
}^{p}%
%TCIMACRO{\dprod \limits_{l=1}^{p-1}}%
%BeginExpansion
{\displaystyle\prod\limits_{l=1}^{p-1}}
%EndExpansion
V_{j_{l}}^{l}\nonumber\\&+&%
%TCIMACRO{\dprod \limits_{l=1}^{k-1}}%
%BeginExpansion
{\displaystyle\prod\limits_{l=1}^{k-1}}
%EndExpansion
V_{j_{l}}^{l}\label{km1}
\cb{\max_{n\in N(j_1,\ldots,j_{k-1})}}\left\{  \sum_{p=k}^{k-1+n}U_{j_{k-1}}^{p}%
%TCIMACRO{\dprod \limits_{l=k}^{p-1}}%
%BeginExpansion
{\displaystyle\prod\limits_{l=k}^{p-1}}
%EndExpansion
V_{j_{k-1}}^{l}+ \IE_{j_{k-1}}Y_{\rho^{j_{k-1}}}^{\ast L-k+1-n}
%TCIMACRO{\dprod \limits_{l=k}^{k-1+n}}%
%BeginExpansion
{\displaystyle\prod\limits_{l=k}^{k-1+n}}
%EndExpansion
V_{j_{k-1}}^{l}\right\}
\end{eqnarray}
\cb{where the maximum runs over the $\mathcal{F}_{j_{k-1}}$-measurable set
$$N(j_1,\ldots,j_{k-1}):=\{n;\;0\leq n\leq (v_{j_{k-1}}-\mathcal{E}_{k-1}(j_{1},...,j_{k-1}
))\wedge (L-k+1)\}.$$}

\end{itemize}

\end{proposition}

\begin{proof}
\cb{For $k=L+1$ both assertions are implied by the conventions $Y_{r}^{\ast 0,j_{1},\ldots,j_{L}}=X_{j_{1},...,j_{L}}$
and $Y^{\ast 0}_r=0$. Hence we assume for the remainder of the proof that $1< k\leq L$.}\\[0.1cm]
(i)
Under $\mathcal{C}_{k-1}(j_{1},...,j_{k-1})=1$, we have for $r>j_{k-1}$
\begin{eqnarray}
&& Y_{r}^{\ast L-k+1,j_{1},\ldots,j_{k-1}} \nonumber\\  &  =&\sum_{p=1}^{k-1}U_{j_{p}}^{p}%
{\displaystyle\prod\limits_{l=1}^{p-1}}
V_{j_{l}}^{l} +%
{\displaystyle\prod\limits_{l=1}^{k-1}}
V_{j_{l}}^{l}
  \underset{\mathcal{C}_{L}(j_{1},...,j_{k-1},\tau^{k}%
,...,\tau^{L})=1}{\,\sup_{r\leq \tau^{k}\leq\cdots\leq\tau^{L}}}%
\IE_{r}\left(  \sum_{p=k}^{L}U_{\tau^{p}}^{p}%
{\displaystyle\prod\limits_{l=k}^{p-1}}
V_{\tau^{l}}^{l}\right)  . \label{sp1}
\end{eqnarray}
\cb{As $r>j_{k-1\text{ }}$,
we obtain, thanks to (\ref{eq:hilf1}),}
\begin{eqnarray}
&& \underset{\mathcal{C}_{L}(j_{1},...,j_{k-1},\tau^{k}%
,...,\tau^{L})=1}{\,\sup_{r\leq \tau^{k}\leq\cdots\leq\tau^{L}}}%
\IE_{r}\left(  \sum_{p=k}^{L}U_{\tau^{p}}^{p}%
{\displaystyle\prod\limits_{l=k}^{p-1}}
V_{\tau^{l}}^{l}\right)  \nonumber \\
&=&\1_{\{ r<\rho^{j_{k-1}} \} }  \underset{\mathcal{C}_{L-k+1}(\tau^{k}%
,...,\tau^{L})=1}{\,\sup_{\rho^{j_{k-1}}\leq \tau^{k}\leq\cdots\leq\tau^{L}}}%
\IE_{r}\left(  \sum_{p=k}^{L}U_{\tau^{p}}^{p}%
{\displaystyle\prod\limits_{l=k}^{p-1}}
V_{\tau^{l}}^{l}\right)  \nonumber\\
&& + \1_{ \{r\geq\rho^{j_{k-1}} \} } \underset{\mathcal{C}_{L-k+1}(\tau^{k}%
,...,\tau^{L})=1}{\,\sup_{r\leq \tau^{k}\leq\cdots\leq\tau^{L}}}%
\IE_{r}\left(  \sum_{p=k}^{L}U_{\tau^{p}}^{p}%
{\displaystyle\prod\limits_{l=k}^{p-1}}
V_{\tau^{l}}^{l}\right)   \nonumber\\
&  =&\1_{ \{ r<\rho^{j_{k-1}} \}}\IE_{r}Y_{\rho^{j_{k-1}}}^{\ast L-k+1}+\1_{ \{r\geq
\rho^{j_{k-1}} \}}Y_{r}^{\ast L-k+1}. \label{sp2}%
\end{eqnarray}

Hence, by combining (\ref{sp1}) and (\ref{sp2}) we get (i).\\
\  \\
(ii) \cb{Given that the first $(k-1)$ rights have been exercised at times $j_1\leq \cdots \leq j_{k-1}$ the number
of the remaining $(L-k+1)$ rights which are also exercised  at time $j_{k-1}$ must be chosen from the $\mathcal{F}_{j_{k-1}}$-measurable set
$N(j_1,\ldots,j_{k-1})=\{n;\;0\leq n\leq (v_{j_{k-1}}-\mathcal{E}_{k-1}(j_{1},...,j_{k-1}
))\wedge (L-k+1)\}$. These are the only choices which obey the volume constraint at time $j_{k-1}$. Hence,
$$
Y_{j_{k-1}}^{\ast L-k+1,j_{1},\ldots,j_{k-1}}=\max_{n\in N(j_1,\ldots,j_{k-1})}
\sup_{\substack{\rho^{j_{k-1}}\leq \tau^{n+k}\leq \cdots\leq \tau^{L} \\ \mathcal{C}_{L-k-n+1}(\tau^{k+n}%
,...,\tau^{L})=1}} \; \IE X_{j_1,\ldots,j_{k-1},\cdots, j_{k-1},\tau^{k+n},\ldots, \tau^{L}},
$$
where the time index $j_{k-1}$ appears $(n+1)$ times in the $n$th term. It then follows, for fixed $n\in N(j_1,\ldots,j_{k-1})$,
}
\begin{eqnarray*}
&& \sup_{\substack{\rho^{j_{k-1}}\leq \tau^{n+k}\leq \cdots\leq \tau^{L} \\ \mathcal{C}_{L-k-n+1}(\tau^{k+n}%
,...,\tau^{L})=1}} \; \IE X_{j_1,\ldots,j_{k-1},\cdots, j_{k-1},\tau^{k+n},\ldots, \tau^{L}}\\
&  =& \sum_{p=1}^{k-1}U_{j_{p}}^{p}%
%TCIMACRO{\dprod \limits_{l=1}^{p-1}}%
%BeginExpansion
{\displaystyle\prod\limits_{l=1}^{p-1}}
%EndExpansion
V_{j_{l}}^{l}+%
%TCIMACRO{\dprod \limits_{l=1}^{k-1}}%
%BeginExpansion
{\displaystyle\prod\limits_{l=1}^{k-1}}
%EndExpansion
V_{j_{l}}^{l}\sum_{p=k}^{k-1+n}U_{j_{k-1}}^{p}%
%TCIMACRO{\dprod \limits_{l=k}^{p-1}}%
%BeginExpansion
{\displaystyle\prod\limits_{l=k}^{p-1}}
%EndExpansion
V_{j_{k-1}}^{l}\\
&&  +%
%TCIMACRO{\dprod \limits_{l=1}^{k-1}}%
%BeginExpansion
{\displaystyle\prod\limits_{l=1}^{k-1}}
%EndExpansion
V_{j_{l}}^{l}%
%TCIMACRO{\dprod \limits_{l=k}^{k-1+n}}%
%BeginExpansion
{\displaystyle\prod\limits_{l=k}^{k-1+n}}
%EndExpansion
V_{j_{k-1}}^{l}\sup_{\substack{\rho^{j_{k-1}}\leq \tau^{n+k}\leq \cdots\leq \tau^{L} \\ \mathcal{C}_{L-k-n+1}(\tau^{k+n}%
,...,\tau^{L})=1}}\IE_{j_{k-1}}\sum
_{p=k+n}^{L}U_{\tau^{p}}^{p}%
%TCIMACRO{\dprod \limits_{l=k+n}^{p-1}}%
%BeginExpansion
{\displaystyle\prod\limits_{l=k+n}^{p-1}}
%EndExpansion
V_{\tau^{p}}^{l}\\
&  =&\sum_{p=1}^{k-1}U_{j_{p}}^{p}%
%TCIMACRO{\dprod \limits_{l=1}^{p-1}}%
%BeginExpansion
{\displaystyle\prod\limits_{l=1}^{p-1}}
%EndExpansion
V_{j_{l}}^{l}+%
%TCIMACRO{\dprod \limits_{l=1}^{k-1}}%
%BeginExpansion
{\displaystyle\prod\limits_{l=1}^{k-1}}
%EndExpansion
V_{j_{l}}^{l}\sum_{p=k}^{k-1+n}U_{j_{k-1}}^{p}%
%TCIMACRO{\dprod \limits_{l=k}^{p-1}}%
%BeginExpansion
{\displaystyle\prod\limits_{l=k}^{p-1}}
%EndExpansion
V_{j_{k-1}}^{l}
  +%
%TCIMACRO{\dprod \limits_{l=1}^{k-1}}%
%BeginExpansion
{\displaystyle\prod\limits_{l=1}^{k-1}}
%EndExpansion
V_{j_{l}}^{l}%
%TCIMACRO{\dprod \limits_{l=k}^{k-1+n}}%
%BeginExpansion
{\displaystyle\prod\limits_{l=k}^{k-1+n}}
%EndExpansion
V_{j_{k-1}}^{l}\IE_{j_{k-1}}Y_{\rho^{j_{k-1}}}^{\ast L-k-n+1},
\end{eqnarray*}
\cb{making again use of (\ref{eq:hilf1}). This implies \eqref{km1}.}

\end{proof}

\subsection{Dual representation based on Doob decompositions}

The goal of this subsection is to \cb{prove and discuss the following simplified version of the dual representation from Theorem \ref{dual}
for multiple stopping problems of the form (\ref{eq:generic}). }

\begin{theorem}
\label{cor:corollary_dual}
Suppose $Y_{i}^{\ast L}$ is given by (\ref{eq:generic}). Then:

\medskip

%The dual representation Theorem \ref{dual}-\textit{\textbf{(ii)}} reads
%\begin{align*}
%Y_{i}^{\ast,L}  &  =\max_{i\leq j_{1}\leq\cdot\cdot\cdot\leq j_{L}\leq
%\partial}\left(  X_{j_{1},\ldots,j_{L}}+\sum_{k=1}^{L}%
%%TCIMACRO{\dprod \limits_{l=1}^{k-1}}%
%%BeginExpansion
%{\displaystyle\prod\limits_{l=1}^{k-1}}
%%EndExpansion
%V_{j_{l}}^{(l)}\times\right. \\
%&  \left.  \left(  M_{j_{k-1}}^{\ast L-k+1}-M_{j_{k}}^{\ast L-k+1}+\IE_{j_{k}%
%}A_{\rho^{(j_{k-1})}}^{\ast L-k+1}-\IE_{j_{k-1}}A_{\rho^{(j_{k-1})}}^{\ast
%L-k+1}\right)  \right)  .
%\end{align*}

\textbf{(i)} For any set of martingales $\left(  M_{r}^{L-k+1}\right)_{r\geq 0}$, $k=1,\ldots,L$, and
any set of \cb{integrable }adapted processes  $\left(  A_{r}^{L-k+1}\right)_{r\geq 0}$, $k=1,\ldots,L$, it holds for $i\geq 0$ and with
$j_0:=i$
\begin{eqnarray*}
Y_{i}^{\ast L}  &  \leq & \IE_i \max_{\substack{i\leq j_{1}\leq\cdot\cdot\cdot\leq j_{L}\leq
\partial\\ \mathcal{C}_L(j_1,...,j_L)=1}}  \Biggl( \sum_{k=1}^{L}U_{j_{k}}^{k}
{\displaystyle\prod\limits_{l=1}^{k-1}}
V_{j_{l}}^{l} \\ &&+\sum_{k=1}^{L}
{\displaystyle\prod\limits_{l=1}^{k-1}}
V_{j_{l}}^{l}
  \left(  M_{j_{k-1}}^{L-k+1}-M_{j_{k}}^{L-k+1}+\IE_{j_{k}%
}A_{\rho^{j_{k-1}}}^{L-k+1}-\IE_{j_{k-1}}A_{\rho^{j_{k-1}}}^{L-k+1}\right) \Biggr).
\end{eqnarray*}

\textbf{(ii)} For every $i\geq 0$ it holds with $j_0:=i$
\begin{eqnarray*}
Y_{i}^{\ast L}  &  = & \max_{\substack{i\leq j_{1}\leq\cdot\cdot\cdot\leq j_{L}\leq
\partial\\ \mathcal{C}_L(j_1,...,j_L)=1}}  \Biggl( \sum_{k=1}^{L}U_{j_{k}}^{k}
{\displaystyle\prod\limits_{l=1}^{k-1}}
V_{j_{l}}^{l} \\ &&+\sum_{k=1}^{L}
{\displaystyle\prod\limits_{l=1}^{k-1}}
V_{j_{l}}^{l}
  \left(  M_{j_{k-1}}^{\ast L-k+1}-M_{j_{k}}^{\ast L-k+1}+\IE_{j_{k}%
}A_{\rho^{j_{k-1}}}^{\ast L-k+1}-\IE_{j_{k-1}}A_{\rho^{j_{k-1}}}^{\ast L-k+1}\right) \Biggr),
\end{eqnarray*}
where $M^{\ast L-k+1}$ , $A^{\ast L-k+1}$ are the martingale part and the predictable part of the Doob decomposition
of the auxiliary Snell envelopes $Y^{\ast L-k+1}$ in (\ref{df}), respectively.
\end{theorem}

\cb{We here recall that the Doob decomposition of $Y^{\ast L-k+1}$ is the unique decomposition of the form
\[
Y_{r}^{\ast L-k+1}=Y_{0}^{\ast L-k+1}+M_{r}^{\ast L-k+1}-A_{r}^{\ast L-k+1},
\]
where the martingale $M_{r}^{\ast L-k+1}$ and the predictable process $A_{r}^{\ast L-k+1}$ start in zero at time zero. In order
to prove Theorem \ref{cor:corollary_dual} }
we need the following auxiliary result.

%%%%%%%%%%%%%%%%%%%%%%%%%%%%%%%
%%%%%% proposition %%%%%%%%%%%%
%%%%%%%%%%%%%%%%%%%%%%%%%%%%%%%
\begin{proposition}\label{prop:temp_prop1}
Under the assumption of Theorem \ref{cor:corollary_dual}, a Doob martingale of $ \IE_{r}Y_{\rho^{j_{k-1}}\vee r}^{\ast L-k+1},
$ say $\overline{M}_{r}^{\ast
L-k+1},$ is determined for $r\geq j_{k-1}$  by%
\[
\overline{M}_{r}^{\ast L-k+1}-\overline{M}_{j_{k-1}}^{\ast L-k+1}=M_{r}^{\ast
L-k+1}-M_{j_{k-1}}^{\ast L-k+1}+\IE_{j_{k-1}}A_{\rho^{j_{k-1}}}^{\ast
L-k+1}-\IE_{r}A_{\rho^{j_{k-1}}}^{\ast L-k+1}.
\]
\end{proposition}
%%%%%%%%%%%%%%%%%%%%%%%%%
%%%%%% proof %%%%%%%%%%%%
%%%%%%%%%%%%%%%%%%%%%%%%%
\begin{proof}
Using  \cb{the Doob decomposition }we may write%
\begin{align}
& \quad\;\IE_{r}Y_{\rho^{j_{k-1}}\vee r}^{\ast L-k+1} \nonumber \\
&=\1_{\{r<\rho^{j_{k-1}} \}}\IE_{r}Y_{\rho^{j_{k-1}}}^{\ast L-k+1}+\1_{\{r\geq
\rho^{j_{k-1}} \} }Y_{r}^{\ast L-k+1} \nonumber\\
&= \1_{\{ j_{k-1}\leq r<\rho^{j_{k-1}} \} }\left(  Y_{j_{k-1}}^{\ast
L-k+1}+M_{r}^{\ast L-k+1}-M_{j_{k-1}}^{\ast L-k+1}-\IE_{r}A_{\rho^{j_{k-1}}%
}^{\ast L-k+1}+A_{j_{k-1}}^{\ast L-k+1}\right) \nonumber\\
& \quad +\1_{\{r\geq\rho^{j_{k-1}} \} }\left(  Y_{j_{k-1}}^{\ast L-k+1}+M_{r}^{\ast
L-k+1}-M_{j_{k-1}}^{\ast L-k+1}-A_{r}^{\ast L-k+1}+A_{j_{k-1}}^{\ast
L-k+1}\right) \nonumber\\
&= Y_{j_{k-1}}^{\ast L-k+1}+M_{r}^{\ast L-k+1}-M_{j_{k-1}}^{\ast L-k+1}\nonumber\\
& \quad -\1_{\{ j_{k-1}\leq r<\rho^{j_{k-1}} \} }\IE_{r}A_{\rho^{j_{k-1}}}^{\ast
L-k+1}-\1_{\{ r\geq\rho^{j_{k-1}} \}}A_{r}^{\ast L-k+1}+A_{j_{k-1}}^{\ast L-k+1}\nonumber\\
&= Y_{j_{k-1}}^{\ast L-k+1} + A_{j_{k-1}}^{\ast L-k+1} -\1_{\{r\geq\rho^{j_{k-1}} \} } \big( A_{r}^{\ast L-k+1} - A_{\rho^{j_{k-1}}}^{\ast L-k+1} \big) \label{eq:tmp002}\\
& \quad + M_{r}^{\ast L-k+1}-M_{j_{k-1}}^{\ast L-k+1} - \IE_{r}A_{\rho^{j_{k-1}}}^{\ast
L-k+1}.\label{eq:tmp004}
\end{align}

Since line \eqref{eq:tmp002} is the sum of a $\cF_{k-1}$-measurable \cb{random variable }and a predictable \cb{process }and
 line \eqref{eq:tmp004} is a martingale, the proposition follows.

\end{proof}
%%%%%%%%%%%%%%%%%%%%%%%
%%%%%% end proof %%%%%%
%%%%%%%%%%%%%%%%%%%%%%%

\cb{We  now can prove the dual representation.}
%%%%%%%%%%%%%%%%%%%%%%%
%%%%%% corollary  %%%%%
%%%%%%%%%%%%%%%%%%%%%%%

\begin{proof}[Proof of Theorem \ref{cor:corollary_dual}]
\textit{\textbf{(i)}}
\cb{Suppose that, for $k=1,\ldots,L$, $M^{L-k+1}$ is a martingale and $A^{L-k+1}$ is an adapted and integrable process. Then, the }
process $M_{r}^{L-k+1,j_{1},\ldots,j_{k-1}}$ defined for $r\geq j_{k-1}$ via
\begin{eqnarray*}
&&\cb{M_{r}^{L-k+1,j_{1},\ldots,j_{k-1}}}:=
 \prod_{l=1}^{k-1}  V_{j_{l}}^{l}\left(  M_{r}^{L-k+1}-M_{j_{k-1}}^{L-k+1}+\cb{\IE_{j_{k-1}}A_{\rho^{j_{k-1}}}^{
L-k+1}-\IE_{r}A_{\rho^{j_{k-1}}}^{L-k+1}}\right)
\end{eqnarray*}
is a martingale \cb{due to the boundedness of the $V^{l}$'s. By Theorem \ref{dual}-\textbf{\textit{(i)}}, we have
\begin{eqnarray*}
Y_{i}^{\ast L}&\leq& \IE_{i}\max_{i\leq j_{1}\leq\cdot\cdot\cdot\leq j_{L}%
\leq\partial}\left(  X_{j_{1},...,j_{L}}+\sum_{k=1}^{L}\left(  M_{j_{k-1}%
}^{L-k+1,j_{1},...,j_{k-1}}-M_{j_{k}}^{L-k+1,j_{1},...,j_{k-1}}\right)
\right) \\ &=& \IE_{i}\max_{i\leq j_{1}\leq\cdot\cdot\cdot\leq j_{L}%
\leq\partial}\Biggl(  X_{j_{1},...,j_{L}}\\&& +\sum_{k=1}^{L}
{\displaystyle\prod\limits_{l=1}^{k-1}}
V_{j_{l}}^{l}
  \left(  M_{j_{k-1}}^{L-k+1}-M_{j_{k}}^{L-k+1}+\IE_{j_{k}%
}A_{\rho^{j_{k-1}}}^{L-k+1}-\IE_{j_{k-1}}A_{\rho^{j_{k-1}}}^{L-k+1}\right)\Biggr)
\end{eqnarray*}
with $X$ as defined in (\ref{eq:tmp018}) for sufficiently large $N\in \mathbb{N}$. Letting $N$ tend to infinity, we observe
that maximization only takes place over those $j_1\leq\cdots\leq j_L$ which satisfy $\mathcal{C}_L(j_1,\ldots,j_L)=1$. Plugging in
the definition of $X$ for those $j_1\leq\cdots\leq j_L$ yields the assertion.}
\bigskip

\textit{\textbf{(ii)}} \cb{We now apply Theorem \ref{dual}-\textbf{\textit{(ii)}} for $X$ as defined in (\ref{eq:tmp018}) with sufficiently large $N\in \mathbb{N}$.
Letting $N$ tend to infinity again and substituting the definition of $X$, we obtain,
\begin{eqnarray*}
Y_{i}^{\ast L}  &  = & \max_{\substack{i\leq j_{1}\leq\cdot\cdot\cdot\leq j_{L}\leq
\partial\\ \mathcal{C}_L(j_1,...,j_L)=1}}  \Biggl( \sum_{k=1}^{L}U_{j_{k}}^{k}
{\displaystyle\prod\limits_{l=1}^{k-1}}
V_{j_{l}}^{l} +\sum_{k=1}^{L}\left(  M_{j_{k-1}}^{\ast
L-k+1,j_{1},\ldots,j_{k-1}}-M_{j_{k}}^{\ast L-k+1,j_{1},\ldots,j_{k-1}%
}\right) \Biggr),
\end{eqnarray*}
whenever $\left(  M_{r}^{\ast L-k+1,j_{1},\ldots,j_{k-1}}\right)  _{r\geq j_{k-1}}$ are
Doob martingales of $\left(  Y_{r}^{\ast L-k+1,j_{1},\ldots,j_{k-1}}\right)
_{r\geq j_{k-1}}$. By Proposition \ref{Prop10}-\textbf{\textit{(i)}} and Proposition \ref{prop:temp_prop1} we can take}
%Applying Theorem \ref{dual}-\textbf{\textsl{(ii)}}, we can write
\begin{eqnarray*}
&&M_{j_{k-1}}^{\ast
L-k+1,j_{1},\ldots,j_{k-1}}-M_{j_{k}}^{\ast L-k+1,j_{1},\ldots,j_{k-1}} \\ &=& {\displaystyle\prod\limits_{l=1}^{k-1}}
V_{j_{l}}^{l}
  \left(  M_{j_{k-1}}^{\ast L-k+1}-M_{j_{k}}^{\ast L-k+1}+\IE_{j_{k}%
}A_{\rho^{j_{k-1}}}^{\ast L-k+1}-\IE_{j_{k-1}}A_{\rho^{j_{k-1}}}^{\ast L-k+1}\right).
\end{eqnarray*}
\end{proof}

\medskip
\cb{Theorem \ref{cor:corollary_dual} gives a straightforward generic way to calculate upper bounds for multiple
stopping problems of the form (\ref{eq:generic}) at time $i=0$ via Monte Carlo by performing the following steps in a Markovian setting:
\begin{enumerate}
 \item Solve the dynamic program in Proposition \ref{prop:dynamic_prog} for the auxiliary problems $Y^{\ast L-k+1}$
 approximately, and let  $\hat Y^{L-k+1},$ $k=1,\ldots, L,$ denote the respective approximations.
\item Perform the Doob decomposition of $\hat Y^{L-k+1}$, $k=1,\ldots, L$, numerically, e.g. by one layer of nested Monte Carlo
as suggested by \citet{AB2004} in the context of options with a single early exercise right.
\item Plug the processes which stem from the numerical Doob decomposition into the formula of Theorem \ref{cor:corollary_dual}-\textit{\textbf{(i)}}
and replace the outer expectation by the sample mean.
\end{enumerate}
This program will be carried out in more detail in Section \ref{sec:num_ex} in the context of Swing options.

Notice that for a large maturity and a large number of exercise rights, the pathwise maximum in the dual
representation of  Theorem \ref{cor:corollary_dual} runs over a huge set. We will now show that, due to the special structure of the payoff in (\ref{eq:generic}),
this maximum can be computed efficiently
by a recursion over the time steps and exercise levels.}

\medskip
Given any $L$-tuple of martingales $M=(M^{1},\ldots, M^{L})$ and any $L$-tuple of
 adapted processes $A=(A^{1},\ldots, A^{L})$, define, for $n=0,\ldots, L$ and $i=0,\ldots,\partial$,
\begin{eqnarray*}
 \theta^{n,L}_i(M,A)&:=&\max_{\substack{j_0=i\leq j_1\leq\cdots\leq j_{L-n} \\ \mathcal{C}_{L-n}(j_{1},...,j_{L-n}) = 1}}
 \sum_{k=1}^{L-n} \left(\prod_{l=1}^{k-1} V^{l+n}_{j_l} \right) \Bigl(U^{n+k}_{j_k}-\Big(M^{L-n-k+1}_{j_k}-M^{L-n-k+1}_{j_{k-1}}\Big) \\
&& +{\1}_{\{k>1\,\,  \wedge \,\, j_k>j_{k-1}\}} (A^{L-k-n+1}_{\rho^{j_{k-1}}} - \IE_{j_{k-1}}A^{L-k-n+1}_{\rho^{j_{k-1}}}) \Bigr).
\end{eqnarray*}
By  Theorem \ref{cor:corollary_dual},
$$
Y^{* L}_0\leq \IE[\theta^{0,L}_0(M,A)]
$$
for any pair of $L$-tuples $(M,A)$, and
$$
Y^{* L}_0 = \theta^{0,L}_0(M^*,A^*)
$$
for an optimal pair of  $L$-tuples $(M^*,A^*)$. \cbx{Generalizing a related formula in \cite{Mah} in the context of
flexible (or chooser) caps, }
the expression  $\theta^{0,L}_0(M,A)$ can be recursively calculated by the following
proposition.
\begin{proposition}
\label{prop:recursive_max}
\cb{For every $L$-tuple of martingales $M=(M^{1},\ldots, M^{L})$ and  $L$-tuple of adapted processes $A=(A^{1},\ldots, A^{L})$
it holds for $i=0,\ldots,T$ and $n=0,\ldots,L$,}
\begin{eqnarray*}
  \theta^{n,L}_i(M,A)&=& \max\Bigl\{ \theta^{n,L}_{i+1}(M,A)-(M^{L-n}_{i+1}-M^{L-n}_i),
\max_{\nu=1,\ldots,v_i\wedge(L-n)}  \sum_{k=1}^\nu  \left(\prod_{l=1}^{k-1} V^{l+n}_{i} \right)U^{n+k}_i \\&&+
\left(\prod_{\lambda=1}^{\nu} V^{\lambda+n}_{i} \right)\Bigl(\theta^{n+\nu,L}_{\rho^i}(M,A)-(M^{L-n-\nu}_{\rho^i}-M^{L-n-\nu}_i)
\\ && +A^{L-n-\nu}_{\rho^i}-\IE_i A^{L-n-\nu}_{\rho^i} \Bigr)  \Bigr\},
 \end{eqnarray*}
with
$$
 \theta^{n,L}_\partial(M,A)=\sum_{k=1}^{L-n} \left(\prod_{l=1}^{\nu} V^{l+n}_{\partial} \right)U^{n+k}_\partial.
$$
\end{proposition}
\begin{proof}
 The formula for $ \theta^{n,L}_\partial(M,A)$ is obvious by definition. In order to prove the recursive formula, we denote
\begin{eqnarray*}
 F^{n,L}(j_0,\ldots, j_{L-n})&=& \sum_{k=1}^{L-n} \left(\prod_{l=1}^{k-1} V^{l+n}_{j_l} \right) \Bigl(U^{n+k}_{j_k}-(M^{L-n-k+1}_{j_k}-M^{L-n-k+1}_{j_{k-1}}) \\
&& +{\1}_{\{k>1\,\, \wedge \,\, j_k>j_{k-1}\}} (A^{L-k-n+1}_{\rho^{j_{k-1}}} - \IE_{j_{k-1}}A^{L-k-n+1}_{\rho^{j_{k-1}}}) \Bigr).
\end{eqnarray*}
Then,
\begin{equation}\label{recursion1}
 \theta^{n,L}_i(M,A)=\max_{\nu=0,\ldots,v_i\wedge(L-n)} \Bigl\{
 \max_{\substack{j_0=\cdots=j_\nu=i< j_{\nu+1}\leq\cdots\leq j_{L-n} \\ \mathcal{C}_{L-n}(j_{1},...,j_{L-n}) = 1}} F^{n,L}(j_0,\ldots, j_{L-n})\Bigr\}.
\end{equation}
For $\nu=0$ we get
\begin{eqnarray}\label{recursion2}
 && \max_{\substack{j_0=i< j_{1}\leq\cdots\leq j_{L-n} \\ \mathcal{C}_{L-n}(j_{1},...,j_{L-n}) = 1}} F^{n,L}(j_0,\ldots, j_{L-n})
\nonumber\\ &=& \max_{\substack{j_0=i+1\leq j_{1}\leq\cdots\leq j_{L-n} \\ \mathcal{C}_{L-n}(j_{1},...,j_{L-n}) = 1}} F^{n,L}(j_0,\ldots, j_{L-n}) -(M^{L-n}_{i+1}-M^{L-n}_i)\nonumber\\ &=&\theta^{n,L}_{i+1}(M,A)-(M^{L-n}_{i+1}-M^{L-n}_i).
\end{eqnarray}
For $\nu>0$ we obtain
\begin{eqnarray*}
 && \max_{\substack{j_0=\cdots=j_\nu=i< j_{\nu+1}\leq\cdots\leq j_{L-n} \\ \mathcal{C}_{L-n}(j_{1},...,j_{L-n}) = 1}} F^{n,L}(j_0,\ldots, j_{L-n}) \\
&=& \sum_{k=1}^\nu  \left(\prod_{l=1}^{k-1} V^{l+n}_{i} \right)U^{n+k}_i \\&& +
\left(\prod_{\lambda=1}^{\nu} V^{\lambda+n}_{i} \right) \max_{\substack{j_\nu=i,\, \rho^i\leq j_{\nu+1}\leq\cdots\leq j_{L-n} \\ \mathcal{C}_{L-n-\nu}(j_{\nu+1},...,j_{L-n}) = 1}}
\sum_{k=\nu+1}^{L-n} \left(\prod_{l=\nu+1}^{k-1} V^{l+n}_{j_l} \right) \Bigl(U^{n+k}_{j_k}-M^{L-n-k+1}_{j_k} \\
&& \quad\quad +M^{L-n-k+1}_{j_{k-1}}+{\1}_{\{ j_k>j_{k-1}\}} (A^{L-k-n+1}_{\rho^{j_{k-1}}}
- \IE_{j_{k-1}}A^{L-k-n+1}_{\rho^{j_{k-1}}}) \Bigr) \\
&=& \sum_{k=1}^\nu  \left(\prod_{l=1}^{k-1} V^{l+n}_{i} \right)U^{n+k}_i
\\ &&+ \left(\prod_{\lambda=1}^{\nu} V^{\lambda+n}_{i} \right)
\left( (A^{L-\nu-n}_{\rho^{i}} - \IE_{i}A^{L-\nu-n}_{\rho^{i}}) - (M^{L-\nu-n}_{\rho^{i}}-M^{L-\nu-n}_i) \right)
 \\&& +
\left(\prod_{\lambda=1}^{\nu} V^{\lambda+n}_{i} \right) \max_{\substack{j_\nu=\rho^i\leq j_{\nu+1}\leq\cdots\leq j_{L-n} \\ \mathcal{C}_{L-n-\nu}(j_{\nu+1},...,j_{L-n}) = 1}}
\sum_{k=\nu+1}^{L-n} \left(\prod_{l=\nu+1}^{k-1} V^{l+n}_{j_l} \right) \Bigl(U^{n+k}_{j_k}-M^{L-n-k+1}_{j_k} \\
&& \quad\quad +M^{L-n-k+1}_{j_{k-1}}+{\1}_{\{k>\nu+1\,\, \wedge \,\, j_k>j_{k-1}\}} (A^{L-k-n+1}_{\rho^{j_{k-1}}} - \IE_{j_{k-1}}A^{L-k-n+1}_{\rho^{j_{k-1}}}) \Bigr) \\
&=& \sum_{k=1}^\nu  \left(\prod_{l=1}^{k-1} V^{l+n}_{i} \right)U^{n+k}_i
\\ &&+ \left(\prod_{\lambda=1}^{\nu} V^{\lambda+n}_{i} \right) \left( (A^{L-\nu-n}_{\rho^{i}} - \IE_{i}A^{L-\nu-n}_{\rho^{i}}) - (M^{L-\nu-n}_{\rho^{i}}-M^{L-\nu-n}_i) \right)
 \\&& +
\left(\prod_{\lambda=1}^{\nu} V^{\lambda+n}_{i} \right) \max_{\substack{j_0=\rho^i\leq j_{1}\leq\cdots\leq j_{L-n-\nu} \\ \mathcal{C}_{L-n-\nu}(j_1,...,j_{L-n-\nu}) = 1}}
\sum_{k=1}^{L-n-\nu} \left(\prod_{l=1}^{k-1} V^{l+n+\nu}_{j_l} \right) \Bigl(U^{n+\nu+k}_{j_k}-M^{L-n-\nu-k+1}_{j_k}\\
&& \quad\quad +M^{L-n-\nu-k+1}_{j_{k-1}} +{\1}_{\{k>1\,\, \wedge \,\, j_k>j_{k-1}\}} (A^{L-k-n-\nu+1}_{\rho^{j_{k-1}}} - \IE_{j_{k-1}}A^{L-k-n-\nu+1}_{\rho^{j_{k-1}}}) \Bigr) \\
&=& \sum_{k=1}^\nu  \left(\prod_{l=1}^{k-1} V^{l+n}_{i} \right)U^{n+k}_i
\\ && + \left(\prod_{\lambda=1}^{\nu} V^{\lambda+n}_{i} \right) \left( (A^{L-\nu-n}_{\rho^{i}} - \IE_{i}A^{L-\nu-n}_{\rho^{i}}) - (M^{L-\nu-n}_{\rho^{i}}-M^{L-\nu-n}_i)\right) \\
&& + \left(\prod_{\lambda=1}^{\nu} V^{\lambda+n}_{i} \right) \theta^{n+\nu,L}_{\rho^i}(M,A).
\end{eqnarray*}
Plugging this identity and (\ref{recursion2}) into (\ref{recursion1}) yields the assertion.
\end{proof}

\medskip

\subsection{Dual representation based on Snell envelopes}
\cb{In this subsection we present a simplified version of the dual representation in Corollary \ref{cor:tmp_001} in terms of
approximate Snell envelopes for the multiple stopping problem of the form (\ref{eq:generic}). It reads as follows.}

%%%%%%%%%%%%%%%%%%%%%%%%%%%%%
%%%%%% proposition %%%%%%%%%%
%%%%%%%%%%%%%%%%%%%%%%%%%%%%%
\begin{theorem}\label{thm:schwer}
\cb{Suppose $Y_{i}^{\ast L}$ is given by (\ref{eq:generic}) for some fixed $0\leq i \leq \partial$. }Let $(Y^{k})_{1\leq k \leq L}$ be any set
 of \cb{integrable }approximations to $(Y^{\ast k})_{1\leq k \leq L}$ defined in \eqref{df}. We then have,
\cb{with the conventions $j_0:=-1$, $\rho^{j_0}:=i$, and $Y^{0}=0$,}
\begin{eqnarray}\label{ohneplus}
&&Y_{i}^{\ast L}-Y_{i}^{L}\nonumber\\&\leq& \IE_{i}\max_{\substack{i\leq j_{1}\leq\cdot
\cdot\cdot\leq j_{L}\leq\partial,\\\mathcal{C}_{L}(j_{1},...,j_{L})=1}}
\quad \sum_{k=1}^{L}
\Biggl\{\sum_{r=\rho^{j_{k-1}}}^{j_{k}-1} {\displaystyle\prod\limits_{l=1}^{k-1}}
V_{j_{l}}^{l}\left(  \IE_{r}Y_{r+1}^{L-k+1}%
-Y_{r}^{L-k+1}\right)
\nonumber\\ && +
\1_{\{j_{k}>j_{k-1} \}}~%
{\displaystyle\prod\limits_{l=1}^{k-1}}
V_{j_{l}}^{l}
  \Bigl(  \max_{1\leq n\leq v_{j_{k}} \wedge (L-k+1)}\Bigl\{  \sum_{p=k}^{k+n-1}U_{j_{k}}^{p}%
%TCIMACRO{\dprod \limits_{l=k}^{p-1}}%
%BeginExpansion
{\displaystyle\prod\limits_{l=k}^{p-1}}
%EndExpansion
V_{j_{k}}^{l}\nonumber\\ &&\quad\quad+%
%TCIMACRO{\dprod \limits_{l=k}^{k+n}}%
%BeginExpansion
{\displaystyle\prod\limits_{l=k}^{k+n-1}}
%EndExpansion
V_{j_{k}}^{l}\IE_{j_{k}}Y_{\rho^{j_{k}}}^{L-k-n+1}\Bigr\}  -Y_{j_{k}%
}^{L-k+1}\Bigr)  \Biggr\}.\nonumber
\end{eqnarray}
Moreover, the righthand side becomes zero if $Y^{k} = Y^{\ast k},$ for $k=1,...,L.$
\end{theorem}
%%%%%%%%%%%%%%%%%%%%%%%%%%%%%
%%%% END THEOREM %%%%%%%%%%%%
%%%%%%%%%%%%%%%%%%%%%%%%%%%%%
\begin{proof}
Suppose $0\leq i\leq \partial$ is fixed and assume that integrable and adapted processes $Y^{L-k+1}$, $k=1,\ldots, L$, are given which we consider as approximations of the
Snell envelopes of the auxiliary multiple stopping problems $Y^{*L-k+1}$.  Following the
 relationships for the Snell envelopes $Y^{*L-k+1}$ and $Y^{*L-k+1, j_1,\ldots, j_{k-1}}$ in Proposition~\ref{Prop10}, we define
 for $k>1$ approximations to $Y^{*L-k+1, j_1,\ldots, j_{k-1}}$ via
\begin{equation}
Y_{r}^{L-k+1,j_{1},\ldots,j_{k-1}}:=\sum_{p=1}^{k-1}U_{j_{p}}^{p}%
%TCIMACRO{\dprod \limits_{l=1}^{p-1}}%
%BeginExpansion
{\displaystyle\prod\limits_{l=1}^{p-1}}
%EndExpansion
V_{j_{l}}^{l}+\IE_{r}Y_{\rho^{j_{k-1}}\vee r}^{L-k+1}%
%TCIMACRO{\dprod \limits_{l=1}^{k-1}}%
%BeginExpansion
{\displaystyle\prod\limits_{l=1}^{k-1}}
%EndExpansion
V_{j_{l}}^{l},\quad  r>j_{k-1}, \label{xp11}%
\end{equation}
and (for $r=j_{k-1}$)
\begin{eqnarray}
&&Y_{j_{k-1}}^{L-k+1,j_{1},\ldots,j_{k-1}} \nonumber\\  &  :=&\sum_{p=1}^{k-1}U_{j_{p}}^{p}%
%TCIMACRO{\dprod \limits_{l=1}^{p-1}}%
%BeginExpansion
{\displaystyle\prod\limits_{l=1}^{p-1}}
%EndExpansion
V_{j_{l}}^{l}\nonumber \\ &&+%
%TCIMACRO{\dprod \limits_{l=1}^{k-1}}%
%BeginExpansion
{\displaystyle\prod\limits_{l=1}^{k-1}}
%EndExpansion
V_{j_{l}}^{l}\max_{n\in N(j_1,\ldots, j_{k-1})}\left\{  \sum_{p=k}^{k-1+n}U_{j_{k-1}}^{p}%
%TCIMACRO{\dprod \limits_{l=k}^{p-1}}%
%BeginExpansion
{\displaystyle\prod\limits_{l=k}^{p-1}}
%EndExpansion
V_{j_{k-1}}^{l}+%
%TCIMACRO{\dprod \limits_{l=k}^{k-1+n}}%
%BeginExpansion
{\displaystyle\prod\limits_{l=k}^{k-1+n}}
%EndExpansion
V_{j_{k-1}}^{l}\IE_{j_{k-1}}Y_{\rho^{j_{k-1}}}^{L-k+1-n}\right\}. \label{xp2}
\end{eqnarray}
Clearly, we define, for $k=1$, $Y^{L,\emptyset}=Y^{L}$.

Applying Corollary~\ref{cor:tmp_001} for $X$ as defined in (\ref{eq:generic}) and the above approximations we obtain,
\begin{eqnarray}
Y_{i}^{\ast L}  &  \leq & Y_{i}^{L}+\IE_{i}\max_{\substack{i\leq j_{1}\leq
\cdot\cdot\cdot\leq j_{L}\leq\partial,\\\mathcal{C}_{L}(j_{1},...,j_{L}%
) = 1}}\sum_{k=1}^{L}\Bigg(  Y_{j_{k}}^{L-k,j_{1}%
,\ldots,j_{k}}-Y_{j_{k}}^{L-k+1,j_{1},\ldots,j_{k-1}} \nonumber\\
&&  +\sum_{l=j_{k-1}}^{j_{k}-1}\left(  \IE_{l}Y_{l+1}^{L-k+1,j_{1}%
,\ldots,j_{k-1}}-Y_{l}^{L-k+1,j_{1},\ldots,j_{k-1}}\right)  \Bigg),
\label{plus1}%
\end{eqnarray}
where we again observe that the pathwise maximum is attained on the set $\mathcal{C}_{L}(j_{1},...,j_{L}%
) = 1$ by letting $N$ (in the definition of $X$) tend to infinity.

In order to prove the upper bound, it is, in view of (\ref{plus1}), sufficient to show that, for $i\leq j_1 \leq \cdots\leq j_L\leq \partial$ with $\mathcal{C}_{L}(j_{1},...,j_{L}%
) = 1$ the following assertions are true: \\
(i) If $k=2,\ldots,L$  and $j_k > j_{k-1}$ or if $k=1$, then
\begin{eqnarray*}
&&  Y_{j_{k}}^{L-k,j_{1},\ldots,j_{k}}-Y_{j_{k}}^{L-k+1,j_{1},\ldots,j_{k-1}}
\\ &=& {\displaystyle\prod\limits_{l=1}^{k-1}}
 V_{j_{l}}^{l}
\left(  \max_{0\leq n\leq (v_{j_{k}}-1)\wedge (L-k)}
\left\{  \sum_{p=k}^{k+n}U_{j_{k}}^{p}%
{\displaystyle\prod\limits_{l=k}^{p-1}}
V_{j_{k}}^{l}+%
{\displaystyle\prod\limits_{l=k}^{k+n}}
V_{j_{k}}^{l}\IE_{j_{k}}Y_{\rho^{j_{k}}}^{L-k-n}\right\}  -Y_{j_{k}%
}^{L-k+1}\right).
\end{eqnarray*}
(ii) If $k=2,\ldots,L$  and $j_k =j_{k-1}$, then
$$
Y_{j_{k}}^{L-k,j_{1},\ldots,j_{k}}-Y_{j_{k}}^{L-k+1,j_{1},\ldots,j_{k-1}} \leq 0.
$$
(iii) If $k=1$  and $i\leq r \leq j_1-1$, or if $k=2,\ldots,L$  and  $\rho^{j_{k-1}}\leq r \leq j_{k}-1$, then
$$
\IE_{r}Y_{r+1}^{L-k+1,j_{1}%
,\ldots,j_{k-1}}-Y_{r}^{L-k+1,j_{1},\ldots,j_{k-1}} = {\displaystyle\prod\limits_{l=1}^{k-1}}
V_{j_{l}}^{l}\left(  \IE_{r}Y_{r+1}^{L-k+1}%
-Y_{r}^{L-k+1}\right).
$$
(iv) For $k=2,\ldots, L$ and $j_{k-1}\leq r < \rho^{j_{k-1}}$
$$
\IE_{r}Y_{r+1}^{L-k+1,j_{1}%
,\ldots,j_{k-1}}-Y_{r}^{L-k+1,j_{1},\ldots,j_{k-1}}\leq 0.
$$

 We first show (i). To this end suppose that $k\geq 2$ and $j_k> j_{k-1}$. Then,
 $\mathcal{E}_k(j_1,\ldots, j_k)=1$, which implies $N(j_1,\ldots,j_k)=
\{n;\; 0\leq n \leq (v_{k}-1) \wedge (L-k)\}$. Hence, by (\ref{xp2}),
\begin{eqnarray*}
 &&Y_{j_{k}}^{L-k,j_{1},\ldots,j_{k}}\\ &=&\sum_{p=1}^{k-1}U_{j_{p}}^{p}%
%TCIMACRO{\dprod \limits_{l=1}^{p-1}}%
%BeginExpansion
{\displaystyle\prod\limits_{l=1}^{p-1}}
%EndExpansion
V_{j_{l}}^{l} +%
%TCIMACRO{\dprod \limits_{l=1}^{k-1}}%
%BeginExpansion
{\displaystyle\prod\limits_{l=1}^{k-1}}
%EndExpansion
V_{j_{l}}^{l}\max_{0\leq n\leq (v_{j_k}-1)\wedge (L-k)}\left\{  \sum_{p=k}^{k+n}U_{j_{k}}^{p}%
%TCIMACRO{\dprod \limits_{l=k}^{p-1}}%
%BeginExpansion
{\displaystyle\prod\limits_{l=k}^{p-1}}
%EndExpansion
V_{j_{k}}^{l}+%
%TCIMACRO{\dprod \limits_{l=k}^{k-1+n}}%
%BeginExpansion
{\displaystyle\prod\limits_{l=k}^{k+n}}
%EndExpansion
V_{j_{k}}^{l}\IE_{j_{k}}Y_{\rho^{j_{k}}}^{L-k-n}\right\}.
\end{eqnarray*}
Subtracting the defining equation (\ref{xp11}) for $Y^{L-k+1,j_1,\ldots, j_{k-1}}_{j_k}$ from the above expression, we obtain (i), because $j_k \geq \rho^{j_{k-1}}$. For $k=1$, we again get $\mathcal{E}_1(j_1)=1$ and (i) follows in the same way, taking the definition $Y^{L,\emptyset}_{j_1}=Y^{L}_{j_1}$ into account.

In order to derive (ii), we note that $\mathcal{E}_k(j_1,\ldots, j_k)=
\mathcal{E}_{k-1}(j_1,\ldots, j_{k-1})+1$ for $j_k=j_{k-1}$. Thus,
\begin{eqnarray*}
 &&Y_{j_{k}}^{L-k,j_{1},\ldots,j_{k}}\\ &=&\sum_{p=1}^{k-1}U_{j_{p}}^{p}%
%TCIMACRO{\dprod \limits_{l=1}^{p-1}}%
%BeginExpansion
{\displaystyle\prod\limits_{l=1}^{p-1}}
%EndExpansion
V_{j_{l}}^{l} \\ &&+%
%TCIMACRO{\dprod \limits_{l=1}^{k-1}}%
%BeginExpansion
{\displaystyle\prod\limits_{l=1}^{k-1}}
%EndExpansion
V_{j_{l}}^{l}\max_{0\leq n\leq (v_{j_k}-\mathcal{E}_{k-1}(j_1,\ldots, j_{k-1})-1)\wedge (L-k)}\left\{  \sum_{p=k}^{k+n}U_{j_{k}}^{p}%
%TCIMACRO{\dprod \limits_{l=k}^{p-1}}%
%BeginExpansion
{\displaystyle\prod\limits_{l=k}^{p-1}}
%EndExpansion
V_{j_{k}}^{l}+%
%TCIMACRO{\dprod \limits_{l=k}^{k-1+n}}%
%BeginExpansion
{\displaystyle\prod\limits_{l=k}^{k+n}}
%EndExpansion
V_{j_{k}}^{l}\IE_{j_{k}}Y_{\rho^{j_{k}}}^{L-k-n}\right\}
\\ &\leq & \sum_{p=1}^{k-1}U_{j_{p}}^{p}%
%TCIMACRO{\dprod \limits_{l=1}^{p-1}}%
%BeginExpansion
{\displaystyle\prod\limits_{l=1}^{p-1}}
%EndExpansion
V_{j_{l}}^{l} \\ &&+%
%TCIMACRO{\dprod \limits_{l=1}^{k-1}}%
%BeginExpansion
{\displaystyle\prod\limits_{l=1}^{k-1}}
%EndExpansion
V_{j_{l}}^{l}\max_{0\leq n\leq (v_{j_{k-1}}-\mathcal{E}_{k-1}(j_1,\ldots, j_{k-1}))\wedge (L-k+1)}\left\{  \sum_{p=k}^{k+n}U_{j_{k-1}}^{p}%
%TCIMACRO{\dprod \limits_{l=k}^{p-1}}%
%BeginExpansion
{\displaystyle\prod\limits_{l=k}^{p-1}}
%EndExpansion
V_{j_{k-1}}^{l}+%
%TCIMACRO{\dprod \limits_{l=k}^{k-1+n}}%
%BeginExpansion
{\displaystyle\prod\limits_{l=k}^{k+n}}
%EndExpansion
V_{j_{k-1}}^{l}\IE_{j_{k-1}}Y_{\rho^{j_{k-1}}}^{L-k-n}\right\}
\\ &=& Y_{j_{k-1}}^{L-k+1,j_{1},\ldots,j_{k-1}} = Y_{j_{k}}^{L-k+1,j_{1},\ldots,j_{k-1}}.
\end{eqnarray*}

We next prove (iii). The case $k=1$ is trivial in view of the definition of
$Y^{L,\emptyset}$. Hence, we assume that $k\geq 2$ and $\rho^{j_{k-1}}\leq r \leq j_{k}-1$. Then, $r+1>r\geq\rho^{j_{k-1}}>j_{k-1}$ and, thus, by (\ref{xp11})
\begin{eqnarray*}
 \IE_r Y_{r+1}^{L-k+1,j_{1},\ldots,j_{k-1}}&=&\sum_{p=1}^{k-1}U_{j_{p}}^{p}%
%TCIMACRO{\dprod \limits_{l=1}^{p-1}}%
%BeginExpansion
{\displaystyle\prod\limits_{l=1}^{p-1}}
%EndExpansion
V_{j_{l}}^{l}+\IE_{r}Y_{ r+1}^{L-k+1}%
%TCIMACRO{\dprod \limits_{l=1}^{k-1}}%
%BeginExpansion
{\displaystyle\prod\limits_{l=1}^{k-1}}
%EndExpansion
V_{j_{l}}^{l}, \\ Y_{r}^{L-k+1,j_{1},\ldots,j_{k-1}}&=&\sum_{p=1}^{k-1}U_{j_{p}}^{p}%
%TCIMACRO{\dprod \limits_{l=1}^{p-1}}%
%BeginExpansion
{\displaystyle\prod\limits_{l=1}^{p-1}}
%EndExpansion
V_{j_{l}}^{l}+Y_{ r}^{L-k+1}%
%TCIMACRO{\dprod \limits_{l=1}^{k-1}}%
%BeginExpansion
{\displaystyle\prod\limits_{l=1}^{k-1}}
%EndExpansion
V_{j_{l}}^{l}. \
\end{eqnarray*}
Taking the difference of both equations yields (iii).

It remains to show (iv). For $k\geq 2$ and $j_{k-1}<r <\rho^{j_{k-1}}$,
(\ref{xp11}) implies
\begin{eqnarray*}
 \IE_r Y_{r+1}^{L-k+1,j_{1},\ldots,j_{k-1}}&=&\sum_{p=1}^{k-1}U_{j_{p}}^{p}%
%TCIMACRO{\dprod \limits_{l=1}^{p-1}}%
%BeginExpansion
{\displaystyle\prod\limits_{l=1}^{p-1}}
%EndExpansion
V_{j_{l}}^{l}+\IE_{r}Y_{\rho^{j_{k-1}}}^{L-k+1}%
%TCIMACRO{\dprod \limits_{l=1}^{k-1}}%
%BeginExpansion
{\displaystyle\prod\limits_{l=1}^{k-1}}
%EndExpansion
V_{j_{l}}^{l}=Y_{r}^{L-k+1,j_{1},\ldots,j_{k-1}}.
\end{eqnarray*}
Finally, for $k\geq 2$ and $r=j_{k-1}$, by (\ref{xp11}) and (\ref{xp2}),
\begin{eqnarray*}
 && \IE_{j_{k-1}} Y_{j_{k-1}+1}^{L-k+1,j_{1},\ldots,j_{k-1}}=
\sum_{p=1}^{k-1}U_{j_{p}}^{p}%
%TCIMACRO{\dprod \limits_{l=1}^{p-1}}%
%BeginExpansion
{\displaystyle\prod\limits_{l=1}^{p-1}}
%EndExpansion
V_{j_{l}}^{l}+\IE_{j_{k-1}}Y_{\rho^{j_{k-1}}}^{L-k+1}%
%TCIMACRO{\dprod \limits_{l=1}^{k-1}}%
%BeginExpansion
{\displaystyle\prod\limits_{l=1}^{k-1}}
%EndExpansion
V_{j_{l}}^{l} \\ &\leq & \sum_{p=1}^{k-1}U_{j_{p}}^{p}%
%TCIMACRO{\dprod \limits_{l=1}^{p-1}}%
%BeginExpansion
{\displaystyle\prod\limits_{l=1}^{p-1}}
%EndExpansion
V_{j_{l}}^{l}\nonumber \\ &&+%
%TCIMACRO{\dprod \limits_{l=1}^{k-1}}%
%BeginExpansion
{\displaystyle\prod\limits_{l=1}^{k-1}}
%EndExpansion
V_{j_{l}}^{l}\max_{n\in N(j_1,\ldots, j_{k-1})}\left\{  \sum_{p=k}^{k-1+n}U_{j_{k-1}}^{p}%
%TCIMACRO{\dprod \limits_{l=k}^{p-1}}%
%BeginExpansion
{\displaystyle\prod\limits_{l=k}^{p-1}}
%EndExpansion
V_{j_{k-1}}^{l}+%
%TCIMACRO{\dprod \limits_{l=k}^{k-1+n}}%
%BeginExpansion
{\displaystyle\prod\limits_{l=k}^{k-1+n}}
%EndExpansion
V_{j_{k-1}}^{l}\IE_{j_{k-1}}Y_{\rho^{j_{k-1}}}^{L-k+1-n}\right\}
\\ &=& Y_{j_{k-1}}^{L-k+1,j_{1},\ldots,j_{k-1}}.
\end{eqnarray*}
Hence, the asserted upper bound for $Y^{\ast L}_i -Y^{L}_i$ is shown. This upper bound is zero, if  $Y^{k} = Y^{\ast k},$ for $k=1,...,L$, because, by
Proposition \ref{prop:dynamic_prog}, $Y^{\ast L-k+1}$ is a supermartingale and $Y^{\ast L-k+1}_{j_k}$ dominates
$$
 \max_{1\leq n\leq v_{j_{k}} \wedge (L-k+1)}\Bigl\{  \sum_{p=k}^{k+n-1}U_{j_{k}}^{p}%
%TCIMACRO{\dprod \limits_{l=k}^{p-1}}%
%BeginExpansion
{\displaystyle\prod\limits_{l=k}^{p-1}}
%EndExpansion
V_{j_{k}}^{l}+%
%TCIMACRO{\dprod \limits_{l=k}^{k+n}}%
%BeginExpansion
{\displaystyle\prod\limits_{l=k}^{k+n-1}}
%EndExpansion
V_{j_{k}}^{l}\IE_{j_{k}}Y_{\rho^{j_{k}}}^{\ast L-k-n+1}\Bigr\}.
$$
\end{proof}

As a spin-off result from Theorem \ref{thm:schwer}, we may write the following upper
bound for $Y_{i}^{\ast,L}$ which avoids the computation of the recursive maximum
from Proposition \ref{prop:recursive_max} (cf. \cite{Schoen2010}[Remark 3.3] \cbx{for a related result in the context
of the standard multiple stopping problem}).
\begin{corollary}
 Suppose all assumptions and all conventions of Theorem \ref{thm:schwer} are in force. Then,
\begin{eqnarray*}
&&Y_{i}^{\ast L}-Y_{i}^{L}\nonumber\\&\leq&  \IE_{i}\Biggl\{  \sum_{r=i}^{T-1}\max_{0\leq k<L}\left(  \mathcal{V}_{\max}%
^{k}\left(  \IE_{r}Y_{r+1}^{L-k}-Y_{r}^{L-k}\right)  ^{+}\right)
+ \sum_{k=1}^{L}\mathcal{V}_{\max}^{k-1}\\ &&\times
  \max_{i\leq j\leq\partial}\left(  \max_{1\leq n\leq v_{j_{k}%
}\wedge(L-k+1)}\left(  \sum_{p=k}^{k+n-1}U_{j}^{p}%
%TCIMACRO{\dprod \limits_{l=k}^{p-1}}%
%BeginExpansion
{\displaystyle\prod\limits_{l=k}^{p-1}}
%EndExpansion
V_{j}^{l}
{\displaystyle\prod\limits_{l=k}^{k+n-1}}
%EndExpansion
V_{j}^{l}\IE_{j}Y_{\rho^{j}}^{L-k-n+1}\right)  -Y_{j}^{L-k+1}\right)
^{+}\Biggr\},
\end{eqnarray*}
where
$$
\mathcal{V}_{\max}^{k}:=
{\displaystyle\prod\limits_{l=1}^{k}}
%EndExpansion
\max_{j\geq i}V_{j}^{l}.
$$
Moreover, the righthand side becomes zero if $Y^{k} = Y^{\ast k},$ for $k=1,...,L.$
\end{corollary}
\begin{proof}
 It is straightforward to check that the upper bound in this corollary is actually an upper bound to the righthand side of the estimate in Theorem \ref{thm:schwer}. That the bound is still tight, i.e. that the righthand side
becomes zero, if $Y^{k} = Y^{\ast k},$ for $k=1,...,L,$ follows from the same argument as at the end of the proof of Theorem \ref{thm:schwer}.
\end{proof}

%%%%%%%%%%%%%%%%%%%%%%%
%%%%%% section %%%%%%%%
%%%%%%%%%%%%%%%%%%%%%%%
\section{A numerical example}\label{sec:num_ex}

We provide a numerical example for the dual representation of multiple stopping problems \cbx{in the context of swing option pricing. }
Throughout this section,
we assume $i=0$, \cbx{i.e. we provide confidence bounds for the swing option price at time 0. Precisely, }we consider a stylized swing option,
similar to those considered in \cite{MH2004} and \cite{Bender2010}.
In our setting, the holder of a swing option has the right to buy a certain quantity of electricity in the
period from $j=0,\dots,T$, for a fixed strike price $K>0$, subject to the restriction that the option
allows up to $L\geq 1$ exercise opportunities under the volume constraints $v_j$, and where a refraction period has to be taken into account. Here we choose
$T=50$ and \cbx{recall that $\partial := T+1$. }The price of electricity, $(S_t)_{t=0,\ldots,T}$, is modeled by the following discretized exponential Gaussian Ornstein-Uhlenbeck process
\begin{align}
\label{eq:spot_price}
\log(S_j) = (1-k) \big( \log(S_{j-1}) - \mu \big) + \mu + \sigma \epsilon_j, ~ S_0 = s_0 > 0,
\end{align}
where $(\epsilon_j)_{j=1,...,T}$ is a family of \cbx{independent }standard normal random variables and the parameters are specified by
\begin{align*}
\sigma = 0.5, ~ k = 0.9, ~ \mu = 0, ~ s_0 = 1.
\end{align*}
\cbx{We set $S_\partial=0$, which means that no penalty is imposed, if the holder of the option does not exercise all rights.}
 The \cbx{payoff of the swing option is then given by $X$ in (\ref{eq:generic}) with }
\begin{align*}
&V^{l}_j := 1, ~ l=1,\ldots,L-1, ~ j=0,\ldots,\cbx{\partial},\\
&U^{p}_j := Z_j := Z(S_j):= (S_j - K)^+, ~ j=0,\ldots,\cbx{\partial}, ~ p=1,\ldots,L.
\end{align*}
\cbx{In our numerical study we assume that the strike price is $K=1$. As volume constraints we consider the situation of a unit volume constraint
$v_i=1$ for $i=0,\ldots, T$ and the situation of an off-peak swing option with $v_i=1$ on weekdays and $v_i=2$ on Saturdays and Sundays.
The refraction period which we impose is a constant refraction period, i.e. $\rho^{i}=(i+\delta)\wedge \partial$ for various choices of the
constant $\delta\in\mathbb{N}$.}

 In this \cbx{Markovian }framework, we produce \cbx{confidence intervals }for the price of the
swing option at time $i=0$ by applying the following steps. \cbx{The procedure below can easily be generalized to the generic
cashflow structure of Section 3, provided the problem has a Markovian structure. (For notational convenience we only spell out the algorithm for the swing option case.)}

\bigskip

\subsection{Implementation}

\textit{Step 1: Precompute \cbx{an approximation of the continuation values}.} ~ We employ least squares Monte Carlo regression to obtain an approximation to the continuation \cbx{values}
\begin{align*}
C^{\ast 1,l}_j (S_j) &:= \E \Big[ Y^{\ast l}_{j+1} \big| \cF_j \Big]\cbx{=\E \Big[ Y^{\ast l}_{j+1} \big| S_{j} \Big]}, ~ C^{\ast 1,l}_T (S_T) = 0,\\
C^{\ast \delta,l}_j (S_j) &:= \E \Big[ Y^{\ast l}_{j+\delta} \big| \cF_j \Big]\cbx{=\E \Big[ Y^{\ast l}_{j+\delta} \big| S_j \Big]}, ~ C^{\ast \delta,l}_T (S_T) = 0,
\end{align*}
with $l=1,\ldots,L$, where here and in the following $j+1$ and $j+\delta$ are to be understood as  $j+1\, \wedge\,\partial$ and $j+\delta\,\wedge\,\partial.$ 
Recall that $(Y^{\ast l}_{j})_{j=0,\ldots,T}$ is \cbx{given by }the dynamic program
from Proposition \ref{prop:dynamic_prog}. We simulate \cbx{$N_1$ independent paths }$(S^m_j)^{m=1,\ldots,N_1}_{j=0,\ldots,T}$. Choosing as basis functions
\begin{align*}
\psi_1(x) := x, \quad \psi_2(x) := (x-K)^+,
\end{align*}
we use $(S^m_j)^{m=1,\ldots,N_1}_{j=0,\ldots,T}$ in a \cbx{straightforward }least squares regression procedure
to \cbx{solve the dynamic program approximately, replacing the conditional expectations by the least squares Monte Carlo estimator. This yields
}approximations to $C^{\ast 1,l}_j(\cdot)$ and $C^{\ast \delta,l}_j(\cdot)$, denoted by $C^{1,l}_j(\cdot)$ and $C^{\delta,l}_j(\cdot)$.
%We conclude this step by discarding all the paths $(S^m_t)^{m=1,\ldots,M}_{t=0,\ldots,T}$.

\bigskip

\textit{Step 2: Compute \cbx{lower bounds.}} ~\cbx{Given the functions $C^{1,l}_j(\cdot)$ and $C^{\delta,l}_j(\cdot)$}, we define
a (suboptimal) stopping rule $\big(\tau^{p,l}_j \big)^{1\leq l \leq L}_{1\leq p \leq l}$ for $0 \leq j \leq T$
\cbx{along a given trajectory $(S_j)_{j=0\ldots,T}$ (which we suppress in the notation below) }using the following iteration. \cbx{Here $\tau^{p,l}_j$ is interpreted as the time at which the investor exercises the $p$th right, if $l$
rights are left at time $j$.}
\begin{align}
\label{eq:tau}
&\tau^{0,l}_j := j-\delta; \nonumber\\
&\cbx{p} := k := 0; \nonumber\\
&while ~ (p < l) \quad do \nonumber\\
&\qquad \tau^{p+1,l}_j := \inf \Big\{( \tau^{p,l}_j + \delta)\wedge \partial \leq r \leq \cbx{\partial}:
\max_{1\leq n\leq v_r \wedge (l-p)} \Big( n Z_r + C^{\delta,l-p-n}_r \Big) \geq C^{1,l-p}_r \Big\},\nonumber\\
&\qquad s:= \tau^{p+1,l}_j, \nonumber\\
&\qquad k := \argmax_{1\leq n\leq v_s \wedge (l-p)} \Big( n Z_s + C^{\delta,l-p-n}_s \Big), \nonumber\\
&\qquad \tau^{p+1,l}_j := \tau^{p+2,l}_j := \ldots := \tau^{p+k,l}_j :=s, \nonumber\\
&\qquad p:= p+k, \\
%&\qquad \tau^{n_1 + 1,l}_t := \inf\Big\{\tau^{n_1,l }_t + \delta \leq r \leq T: \max_{1\leq n \leq v_r \wedge (l-n_1+1)} \Big( nZ_r + \widehat{C}^{l-n_1+1 -n}_r \Big) \geq Y^{l-n_1+1}_r   \Big\}, \nonumber\\
%&\quad\vdots\nonumber\\
%&\tau^{l,l}_t := \begin{cases} \inf\Big\{\tau^{l-1,l }_t + \delta \leq r \leq T: Z_r \geq Y^{1}_r   \Big\},
%									\end{cases}
&end \nonumber
\end{align}

\cbx{When $C^{1,l}_j(\cdot)$ and $C^{\delta,l}_j(\cdot)$ are replaced by $C^{\ast 1,l}_j(\cdot)$ and $C^{\ast \delta,l}_j(\cdot)$, then this
family of stopping times is optimal. Hence, $\big(\tau^{p,l}_j \big)^{1\leq l \leq L}_{1\leq p \leq l}$ is a good family of stopping times,
if the approximations of the continuation values in Step 1 are reasonably close to the true continuation values.}
%%%%%%%%%%%%%%%%%%%%%%%%%%%%
%%%%%% EXAMPLE %%%%%%%%%%%%%%%%%
%%%%%%%%%%%%%%%%%%%%%%%%%%%%%%%%
\medskip
\cbx{
\begin{remark}
(i) In the situation of unit volume constraint (i.e. $v\equiv 1$), the stopping rule \eqref{eq:tau} simplifies to
$\tau^{0,l}_j=j-\delta$ and
\begin{align*}
\tau^{p,l}_j = \inf \big\{(\tau^{p-1,l}_j + \delta)\wedge \partial \leq r \leq \partial: Z_r + C^{\delta,l-p}_r \geq C^{1,l-p+1}_r \big\}, \quad 1\leq l \leq L, \quad 1 \leq p \leq l,
\end{align*}
compare with Eq. (3.7) in \cite{Bender2010}.
\\[0.1cm]
(ii) In the situation of a trivial refraction period (i.e. $\delta=1$), the above construction of approximate stopping rules is also
used in \cite{AH2010}.
\end{remark}
}
\medskip

Setting
\begin{align*}
\underline{Y}^l_0 := \IE_0 \sum_{p=1}^l Z_{\tau^{p,l}_0}, ~ \underline{Y}^l_1 := \IE_1 \sum_{p=1}^l Z_{\tau^{p,l}_1}, ~ \underline{Y}^l_\delta := \IE_\delta \sum_{p=1}^l Z_{\tau^{p,l}_\delta},
\end{align*}
we have that $\underline{Y}^l_0$ is a lower bound for $Y^{\ast l}_0$. By the tower property of the conditional expectation, we also have
\begin{align*}
\IE_0\underline{Y}^l_1 = \IE_0 \sum_{p=1}^l Z_{\tau^{p,l}_1}, \quad \IE_0\underline{Y}^l_\delta = \IE_0 \sum_{p=1}^l Z_{\tau^{p,l}_\delta}.
\end{align*}
As for simulations, we generate a new set of $N_2$ independent paths of the underlying price process, \cbx{which we again denote,
in abuse of notation, by }$(S^m_j)^{m=1,\ldots,N_2}_{j=0,\ldots,T}$. \cbx{Along theses $N_2$ trajectories we compute $\tau^{p,l}_0$ and apply
the notation }
\begin{align*}
\tau^{p,l,m}_0, \quad 1 \leq p \leq l, ~ 1\leq m \leq N_2.
%\tau^{p,l}_0:= \inf\Big\{\tau^{p-1,l}_0 < r \leq T: Z(S_r) + C^{l-p}_r(S_r) \geq C^{l-p+1}_{r}(S_r) \Big\}, ~ 1\leq p \leq l, ~ 1\leq l \leq L,
\end{align*}
Now the lower biased estimate \cbx{$\widehat{\underline{Y}}^l_0$ }for $Y^{\ast l}_0$ is calculated by averaging over the
 $N_2$ realizations of $\sum_{p=1}^l Z_{\tau^{p,l}_0}$, \cbx{i.e. }
\begin{align}
\widehat{\underline{Y}}^l_0 = \frac{1}{N_2}\sum_{m=1}^{N_2} \sum_{p=1}^l Z(S^m_{\tau^{p,l,m}_0}), \quad 1 \leq l \leq L.\label{eq:lower_zero}
\end{align}
\cbx{Similarly, we also construct approximations
$$
\hat \IE_0\underline{Y}^l_1 = \frac{1}{N_2}\sum_{m=1}^{N_2}  \sum_{p=1}^l Z(S_{\tau^{p,l,m}_1}), \quad \hat \IE_0\underline{Y}^l_\delta = \frac{1}{N_2}\sum_{m=1}^{N_2}  \sum_{p=1}^l Z(S_{\tau^{p,l,m}_\delta})
$$
of $\IE_0\underline{Y}^l_1 $ and $\IE_0\underline{Y}^l_\delta$, which we store for later use.}

For constructing confidence intervals, we also save the \cbx{empirical }standard deviation $\std(\widehat{\underline{Y}}^l_0)$.

\bigskip

\textit{Step 3: Compute approximation\cbx{s }to the Snell envelope\cbx{s}.} ~ Using the stopping rule \eqref{eq:tau}, we \cbx{consider }a family of random variables
\begin{align}\label{eq:Y_lower}
\underline{Y}^l_j &:= \IE_j \sum_{p=1}^l Z_{\tau^{p,l}_j}, ~ 1 \leq l \leq L, ~ 0 \leq j \leq T,
\end{align}
which is an approximation to the Snell envelope $\Big(Y^{\ast l}_j \Big)^{1\leq l \leq L}_{0 \leq j \leq T}$. We apply the following procedure to simulate $\underline{Y}^l_j$:

\smallskip

We simulate a new set of $N_3$ paths of the underlying $(S^m_j)^{1\leq m \leq N_3}_{0\leq j \leq T}$ (\cbx{abusing the notation, again}). \cbx{We refer
to these paths as the outer paths. We now fix a pair $(m,j)$ and compute approximations of  $\underline{Y}^l_j$, $\E_j\underline{Y}^l_{j+1}$, and
$\E_j\underline{Y}^l_{j+\delta}$ along the $m$th outer path which are denoted by $\widehat{\underline{Y}}^{l,m}_j$, $\widehat{\E}^m_j\underline{Y}^l_{j+1}$, and
$\widehat{\E}^m_j\underline{Y}^l_{j+\delta}$, respectively. In these approximations the conditional expectations are replaced by the sample mean
over a set of inner simulations. Hence, for the fixed path $S^m$ and the fixed time point $j$, we generate
$N_4$ independent sample paths of $(S_r)_{r=j,\ldots,T}$ under the
conditional law given that $S_j=S^m_j$. These inner paths are denoted by $(\bar S^{\nu}_r)_{r=j,\ldots,T}^{\nu=1,\ldots, N_4}$, suppressing
here and in the following the dependence on $(m,j)$.
 Along
the inner paths $\bar S^{\nu}$  we compute the stopping times  $\tau^{p,l}_i$ for $i=j, j+1,j+\delta$ in \eqref{eq:tau} and apply the notation
\begin{align*}
\tau^{p,l,\nu}_i, ~ 1 \leq p \leq l, ~ 1 \leq l \leq L, ~ \nu=1,\ldots, N_4.
\end{align*}
We now define
\begin{align*}
\widehat{\underline{Y}}^{l,m}_j := \widehat{\IE}^m_j \sum_{p=1}^l Z_{\tau^{p,l}_j} := \frac{1}{N_4}\sum_{\nu=1}^{N_4}
\sum_{p=1}^l  Z\big( \bar S^{\nu}_{\tau^{p,l,\nu}_j}\big).
\end{align*}
 Similarly, we approximate $\IE_j \underline{Y}^{l}_{j+1}$ for the fixed $j$ along the fixed $m$th outer path by
\begin{align*}
\widehat{\IE}^m_j \underline{Y}^{l}_{j+1} := \frac{1}{N_4} \sum_{\nu=1}^{N_4} \sum_{p=1}^l Z\big( \bar S^{\nu}_{\tau^{p,l,\nu}_{j+1}} \big),
\end{align*}
taking the tower property of the conditional expectation into account.
The approximation $\widehat{\IE}^m_j \underline{Y}^{l}_{j+\delta}$ is obtained analogously.}

\smallskip

%%%%%%%%%%%%%%%%%%%%%%%%%%
%%%%%%%%%% REMARK %%%%%%%%
%%%%%%%%%%%%%%%%%%%%%%%%%%
\begin{remark}\label{rem:var_red}
Note that, \cbx{for $j=0$ approximations $\widehat{\underline{Y}}^{l}_0$, $\widehat{\IE}_0 \underline{Y}^{l}_{1}$,
 and  $\widehat{\IE}_0 \underline{Y}^{l}_{\delta}$ of $\underline{Y}^{l}_0$, $\IE_0 \underline{Y}^{l}_{1}$, and  $\IE_0 \underline{Y}^{l}_{\delta}$
were already obtained based on the $N_2$-samples in Step 2. As typically $N_2>N_4$ these approximations are more accurate. Hence, one can perform
Step 3 for $j\geq 1$ only and set }
$$\widehat{\underline{Y}}^{l,m}_0 := \widehat{\underline{Y}}^l_0,\quad
\widehat{\IE}^m_0 \underline{Y}^{l}_{1}:=\widehat{\IE}_0 \underline{Y}^{l}_{1}, \quad
\widehat{\IE}^m_0 \underline{Y}^{l}_{\delta}:=\widehat{\IE}_0 \underline{Y}^{l}_{\delta},  \quad m=1,\ldots,M.$$
This trick of \cbx{applying the more accurate non-nested Monte Carlo simulation of Step $2$ at time 0 }leads to a significant decrease of the variance in the simulation of the upper bound. This is in the same spirit as the computation of low variance upper bounds for the standard stopping problem from \cite{AB2004}).
\end{remark}
%%%%%%%%%%%%%%%%%%%%%%%%%%

\bigskip

\textit{Step 4: Compute the upper bounds.} ~ The Doob decomposition of $\underline{Y}^l_j$ yields the pair $(\underline{M}^l_j,\underline{A}^l_j)$. Note that due to
\begin{align*}
\underline{M}^l_{i+1} - \underline{M}^l_i = \underline{Y}^l_{i+1} - \IE_i \underline{Y}^l_{i+1},
\end{align*}
and
\begin{align*}
-(\underline{M}^l_{i+\delta} - \underline{M}^l_i) + \underline{A}^l_{i+\delta} - \IE_i \underline{A}^l_{i+\delta} = \IE_i \underline{Y}^l_{i+\delta} - \underline{Y}^l_{i+\delta},
\end{align*}
we can rewrite \cbx{the recursion formula in }Proposition \ref{prop:recursive_max} as
\begin{align*}
\theta^{n,L}_i = \max \left\{ \theta^{n,L}_{i+1} + \IE_i \underline{Y}^{L-n}_{i+1} - \underline{Y}^{L-n}_{i+1}, \max_{1 \leq \nu \leq v_i \wedge (L-n)} \Big( \nu Z_i + \theta^{n+\nu,L}_{i+\delta} + \IE_i \underline{Y}^{L-n-\nu}_{i+\delta} - \underline{Y}^{L-n-\nu}_{i+\delta}  \Big)   \right\}.
\end{align*}
\cbx{We now introduce approximations $\theta^{n,L,m}_i$ of  $\theta^{n,L}_i$ along the $m$th outer path of Step 3 by replacing
$\underline{Y}^l_j$, $\E_j\underline{Y}^l_{j+1}$, and
$\E_j\underline{Y}^l_{j+\delta}$ with their simulated counterparts $\widehat{\underline{Y}}^{l,m}_j$, $\widehat{\E}^m_j\underline{Y}^l_{j+1}$, and
$\widehat{\E}^m_j\underline{Y}^l_{j+\delta}$ constructed in Step 3.}

As simulation based estimate for the upper bound, we use $$Y^{up,L}_0 := \frac{1}{N_3} \sum_{m=1}^{N_3} \theta^{0,L,m}_0.$$
\cbx{Replacing the conditional expectations by the sample mean in $\theta^{0,L,m}_0$ introduces an additional bias up thanks to Jensen's
inequality and the convexity of the maximum. Hence, the estimator $Y^{up,L}_0$ is biased up by Theorem \ref{cor:corollary_dual} and
Proposition \ref{prop:recursive_max}.}

 \cbx{Finally, a }$95\%$ confidence interval \cbx{on the price of the swing option }is given by $$\left[ \widehat{\underline{Y}}^L_0 - 1.96\times\std(\widehat{\underline{Y}}^L_0), Y^{up,L}_0 + 1.96\times\std(Y^{up,L}_0) \right].$$

\medskip

\begin{table}[htbp]
\begin{tabular}{l || l l l l | l	l	l l}
\hline
 &	&	& 	&$95\%$ confidence  	&	&	&	&$95\%$ confidence\\
$\delta$  &$L$ &$\widehat{\underline{Y}}^L_0$ &$Y^{up,L}_0$  &interval	&$L$ &$\widehat{\underline{Y}}^L_0$ &$Y^{up,L}_0$  &interval\\
\hline \hline
1		&2	&3.3116	&3.3211	&[3.30738, 3.32229] &3	&4.53627	&4.54806	&[4.53118, 4.54938]	\\
2		&2	&3.27513	&3.28469	&[3.27094, 3.28587]	&3	&4.43753	&4.45154	&[4.43252, 4.45295]	\\
3		&2	&3.2525	&3.26286	&[3.2483, 3.26414]	&3	&4.36706	&4.38245	&[4.36204, 4.38392]\\
4		&2	&3.2313 &3.24083	&[3.22716, 3.242]	&3	&4.29996 &4.31656	&[4.29502, 4.31813]\\
5		&2	&3.20906 &3.22061	&[3.20496, 3.22199]	&3	&4.29996 &4.31656	&[4.29502, 4.31813]\\
6		&2	&3.18613 &3.19809	&[3.18197, 3.19948]	&3	&4.15557 &4.17514	&[4.15063, 4.17697]	\\
8		&2	&3.13625 &3.14984	&[3.13213, 3.15143]	&3	&3.99773 &4.01954	&[3.99289, 4.02158]	\\
10		&2	&3.09022 &3.10332	&[3.08613, 3.1048]	&3	&3.83377 &3.8528	&[3.82898, 3.85464]\\
12		&2	&3.03874 &3.05196	&[3.03468, 3.05356]	&3	&3.65492 &3.67658	&[3.65023, 3.67868]\\
14		&2	&2.98727 &3.00048	&[2.98321, 3.00199]	&3	&3.47017 &3.49061	&[3.46558, 3.49258]\\
16		&2	&2.92751 &2.94214	&[2.9235, 2.9438]	&3	&3.27524 &3.29482	&[3.27077, 3.29674]\\
18		&2	&2.87368 &2.8888	&[2.86964, 2.89049]	&3	&3.09209 &3.11002	&[3.08775, 3.11186]\\
20		&2	&2.81521 &2.83005	&[2.81123, 2.83173]	&3	&2.91951 &2.93649	&[2.91536, 2.9383]\\
\hline
\end{tabular}
\caption{\textit{Unit volume constraints ($v_j \equiv 1$).} Numerical results based on the approximation
$\widehat{\underline{Y}}^{l}_j$ to the Snell envelope via the stopping rule \eqref{eq:tau} for two and three exercise rights.}
\label{table:prag_upper}
%compare also with \cite[Section3.3]{Bender2010} for the case of refraction period $1$.
\end{table}
%%%%%%%%%%%%%%%%%%%%%%%%%%%%%%%%%%%
%%%%% end tabular %%%%%%%%%%%%%%%%%
%%%%%%%%%%%%%%%%%%%%%%%%%%%%%%%%%%%

\medskip

%%%%%%%%%%%%%%%%%%%%%%%%%%%%%%%%%%%%%
%%%%%%%%%%% FIGURE %%%%%%%%%%%%%%%%%%
%%%%%%%%%%%%%%%%%%%%%%%%%%%%%%%%%%%%%
%\begin{figure}[htbp]
%\centering
%\includegraphics[width=8.6cm,height=7.5cm]{Y0_pureRefrac_10}\includegraphics[width=8.6cm,height=7.5cm]{Y0_pureRefrac_20}
%\caption{\textit{Unit volume constraints.} Plots of the upper bounds $\widehat{\underline{Y}}^L_0 + \widehat{\Delta}^{L,\delta}_0$, $\widehat{\underline{Y}}^L_0 + \widehat{\Delta}^{L,\delta;1}_0$ and $\widehat{\underline{Y}}^L_0$. The left figure depicts the case $\delta=10$, the right figure depicts the case $\delta=20$.}\label{pic:pureRefrac}
%\end{figure}
%%%%%%%%%%%%%%%%%%%%%%%%%%%%%%%%%%%%%
%%%%%% END FIGURE %%%%%%%%%%%%%%%%%%%
%%%%%%%%%%%%%%%%%%%%%%%%%%%%%%%%%%%%%

\medskip

\subsection{Numerical results: swing options with unit volume constraints}
%\textbf{\textit{Case 1: unit volume constraints} }

\smallskip

\cbx{We now present some numerical results which the above algorithm produces  for the swing option contract
as specified at the beginning
of this section.
Let us first consider the situation of a unit volume constraint, i.e. $v_j := 1$ for $j=0,\ldots,T$.
We recall that $\delta \in \mathbb{N}$ denotes a constant refraction period. In this setting the dual representation of
 Theorem \ref{cor:corollary_dual} reduces to the one derived in \cite{Bender2010}. In the latter paper the same swing option example is treated
numerically but for up to three exercise rights only. Thanks to the new recursion formula in Proposition \ref{prop:recursive_max} we can
now efficiently treat the case of a large number of exercise rights (here up to $L=10$). Moreover, the upper bound algorithm in \cite{Bender2010}
differs slightly from the one we propose here. In \cite{Bender2010} the upper bound is calculated based on the numerical Doob decomposition
of
$$
{Y}^l_j=\max\{Z_j+ C^{\delta,l-1}_j, C^{1,l}_j\}
$$
while we here utilize the numerical Doob decomposition of  $\underline{Y}^l_j := \IE_j \sum_{p=1}^l Z_{\tau^{p,l}_j}.$}

 The choice of simulation parameters in our study is as follows:
in Step 1, we choose $N_1=1000$ paths for the \cbx{least squares Monte Carlo regression to approximate the }continuation function.
In Step 2, the lower bound is simulated using $N_2=300000$ paths and in Step 3, we employ $N_3=2000$ outer and $N_4=100$ inner paths
\cbx{for the computation of the upper bound. }Moreover, we use the variance reduction method from Remark \ref{rem:var_red}.

Table \ref{table:prag_upper} depicts the numerical results for \cbx{the case of two and three exercise rights for a refraction period
ranging from 1 to 20. We observe that the relative length of the 95\%-confidence intervals is less than 1\%
in all cases.  A comparison with the numerical results in \cite{Bender2010} shows that the
differences in the upper price estimator based on $Y^l$ and $\underline Y^l$ are negligible, but the variance reduction method of Remark
\ref{rem:var_red} shrinks the confidence interval significantly.

The numerical results for the case of a larger number of exercise rights ($L=4,6,8,10$) are presented in Table \ref{table:prag_upper_2}.
Due to the time horizon of 50 days, it may happen that, for a large number of rights and a large refraction period, some exercise rights
cannot be used by the investor. This explains why e.g. the price bounds for the swing option with refraction period $\delta=14$ are
the same for $L=4$ and $L=6$ rights. Concerning the accuracy of our numerical procedure we emphasize that the relative
difference between lower and upper bound is still less than 1\% even in the case of 10 exercise rights.}

\medskip

%%%%%%%%%%%%%%%%%%%%%%%%%%%%%%%%%%%
%%%%%%%%% tabular %%%%%%%%%%%%%%%%%
%%%%%%%%%%%%%%%%%%%%%%%%%%%%%%%%%%%
\begin{table}[htbp]
\begin{tabular}{l || l l l l | l	l	l l}
\hline
 &	&	& 	&$95\%$ confidence  	&	&	&	&$95\%$ confidence\\
$\delta$  &$L$ &$\widehat{\underline{Y}}^L_0$ &$Y^{up,L}_0$  &interval	&$L$ &$\widehat{\underline{Y}}^L_0$ &$Y^{up,L}_0$  &interval\\
\hline \hline
1		&4	&5.60136	&5.614	&[5.59554, 5.61527]	&6	&7.38677	&7.40107	&[7.37977, 7.4023]\\
2		&4	&5.41347	&5.43091	&[5.4078, 5.43249]	&6		&6.94554	&6.97364	&[6.93882, 6.97562]\\
3		&4	&5.27248	&5.29342	&[5.2668, 5.29509]	&6		&6.58739	&6.62054	&[6.58071, 6.62272]\\
4		&4	&5.13119 &5.15543	&[5.12562, 5.15741]	&6		&6.2151 &6.25165	&[6.2086, 6.2541]\\
5		&4	&4.98353 &5.01117	&[4.97802, 5.0133]	&6		&5.81445 &5.85591	&[5.80811, 5.85862]\\
6		&4	&4.82479 &4.85239	&[4.81928, 4.85454]	&6		&5.4079 &5.44248	&[5.40174, 5.44491]\\
8		&4	&4.49057 &4.51822	&[4.48525, 4.52041]	&6		&4.68469 &4.71574	&[4.67908, 4.71808]\\
10		&4	&4.13658 &4.16231	&[4.13141, 4.16444]	&6		&4.1662 &4.19164	&[4.16098, 4.19373]\\
12		&4	&3.78981 &3.81429	&[3.78491, 3.81652]	&6		&3.78992 &3.81448	&[3.78502, 3.81671]\\
14		&4	&3.50023 &3.52138	&[3.49558, 3.52334]	&6		&3.50023 &3.52138	&[3.49558, 3.52334]\\
\hline
1		&8	&8.83286	&8.84907	&[8.82488, 8.85034]	&10		&10.0219	&10.0391	&[10.0131, 10.0404]\\
2		&8	&8.04508	&8.08421	&[8.03754, 8.08651]	&10		&8.80264	&8.85093	&[8.79443, 8.85353]\\
3		&8	&7.36943	&7.41474	&[7.36205, 7.41734]	&10		&7.73096	&7.78117	&[7.72312, 7.78394]\\
4		&8	&6.66726 &6.71129	&[6.66021, 6.71389]	&10		&6.74649 &6.78596	&[6.73929, 6.78847]\\
5		&8	&5.99864 &6.0388	&[5.99202, 6.0414]	&10		&6.0035 &6.04305	&[5.99687, 6.04558]\\
6		&8	&5.45188 &5.48518 &[5.44563, 5.48748]	&10		&5.45187	&5.48518	&[5.44563, 5.48748]\\
\hline
\end{tabular}
\caption{\textit{Unit volume constraints ($v_j \equiv 1$).} Numerical results based on the approximation $\widehat{\underline{Y}}^{l}_j$ to the Snell envelope via the stopping rule \eqref{eq:tau} for a higher number of exercise rights.}
\label{table:prag_upper_2}
%compare also with \cite[Section3.3]{Bender2010} for the case of refraction period $1$.
\end{table}
%%%%%%%%%%%%%%%%%%%%%%%%%%%%%%%%%%%
%%%%% end tabular %%%%%%%%%%%%%%%%%
%%%%%%%%%%%%%%%%%%%%%%%%%%%%%%%%%%%

\medskip

%%%%%%%%%%%%%%%%%%%%%%%%%%%%%%%%%%%%%
%%%%%%%%%%% FIGURE %%%%%%%%%%%%%%%%%%
%%%%%%%%%%%%%%%%%%%%%%%%%%%%%%%%%%%%%
%\begin{figure}[htbp]
%\centering
%\includegraphics[width=8.6cm,height=7.5cm]{Y0_vol_refrac_10}\includegraphics[width=8.6cm,height=7.5cm]{Y0_vol_refrac_20}
%\caption{\textit{Off peak volume constraints.} Plots of the upper bounds $\widehat{\underline{Y}}^L_0 + \widehat{\Delta}^{L,\delta}_0$, $\widehat{\underline{Y}}^L_0 + \widehat{\Delta}^{L,\delta;1}_0$ and $\widehat{\underline{Y}}^L_0$. The left figure depicts the case $\delta=10$, the right figure depicts the case $\delta=20$.}\label{pic:vol_refrac}
%\end{figure}
%%%%%%%%%%%%%%%%%%%%%%%%%%%%%%%%%%%%%
%%%%%% END FIGURE %%%%%%%%%%%%%%%%%%%
%%%%%%%%%%%%%%%%%%%%%%%%%%%%%%%%%%%%%

\medskip

\subsection{Numerical results: off-peak swing option}

%\textbf{\textit{Case 2: Off peak swing option}}

%\smallskip

We now consider a swing option which allows for buying at most one package of electricity on weekdays and \cbx{two }packages on
Saturdays and Sundays (off-peak period). Hence, we have for $j=0,\ldots,50$ the volume constraints
\begin{align}\label{eq:off_peak}
v_j := \begin{cases}
				1,	&\text{if $j$ is a week day,}\\
				2,	&\text{if $j$ is a weekend day,}
			 \end{cases}
\end{align}
where we start in $j=0$ on a Monday.

%%%%%%%%%%%%%%%%%%%%%%%%%%%%%%%%%%%
%%%%%%%%% tabular %%%%%%%%%%%%%%%%%
%%%%%%%%%%%%%%%%%%%%%%%%%%%%%%%%%%%
\begin{table}[htbp]
\begin{tabular}{l | l l l l | l l l l}
\hline
 &	&	& 	&$95\%$ confidence &	&	&	&$95\%$ confidence\\
$\delta$  &$L$ &$\widehat{\underline{Y}}^L_0$ &$Y^{up,L}_0$  &interval	&$L$	&$\widehat{\underline{Y}}^L_0$	 &$Y^{up,L}_0$	&interval\\
\hline \hline
1		&2	&3.39804	&3.40779	&[3.39342, 3.409] &3	&4.72241	&4.73543	&[4.71667, 4.73682]	\\
2		&2	&3.36368	&3.37544	&[3.35908, 3.37676]	&3	&4.63568	&4.65304	&[4.62998, 4.65468]	\\
3		&2	&3.34728	&3.35873	&[3.34265, 3.36017]	&3	&4.58105	&4.59793	&[4.57532, 4.59971]\\
4		&2	&3.32763 &3.34072	&[3.32302, 3.34223]	&3	&4.52367 &4.54475	&[4.51799, 4.5467]\\
5		&2	&3.30829 &3.32101	&[3.30366, 3.32246]	&3	&4.46437 &4.48607	&[4.45867, 4.48806]	\\
6		&2	&3.28626 &3.29915	&[3.28162, 3.30067]	&3	&4.40272 &4.42552	&[4.39701, 4.42758]\\
8		&2	&3.24383 &3.26104	&[3.23921, 3.26284]	&3	&4.29259 &4.31753	&[4.2869, 4.31985]\\
10		&2	&3.20872	&3.22244	&[3.20409, 3.22404]	&3	&4.18401 &4.2091	&[4.1783, 4.21161]\\
12		&2	&3.16994 &3.18626	&[3.16529, 3.18807]	&3	&4.06657 &4.09375	&[4.06086, 4.09637]\\
14		&2	&3.12722 &3.14467	&[3.12253, 3.14666]	&3	&3.95351 &3.97621	&[3.94778, 3.97861]\\
16		&2	&3.08634 &3.10272	&[3.08169, 3.10464]	&3	&3.86402 &3.88597	&[3.85829, 3.88827]\\
18		&2	&3.05186	&3.06793	&[3.04713, 3.06991]	&3	&3.75608 &3.77801	&[3.75036, 3.78047]\\
20		&2	&3.01557 &3.03044	&[3.01088, 3.03233]	&3	&3.65138 &3.67244	&[3.64572, 3.67486]\\
\hline\hline
1		&4	&5.89744	&5.91312	& [5.89078, 5.91464]	&6	&7.91351	&7.93364	&[7.90538, 7.93527]\\
2		&4	&5.73736	&5.76003	&[5.73078, 5.76192]	&6 &7.55394	&7.58547	&[7.54595, 7.58773]\\
3		&4	&5.62688	&5.65097	&[5.62027, 5.65305]	&6	&7.28486	&7.32349	&[7.27684, 7.32611]\\
4		&4	&5.51763 &5.5468	&[5.51105, 5.54915]	&6	&7.01995 &7.06275	&[7.01198, 7.06577]\\
5		&4	&5.40154 &5.43295	&[5.39496, 5.43546]	&6	&6.7418 &6.78724	&[6.73383,6.79042]\\
6		&4	&5.27976 &5.30967	&[5.27316, 5.31227]	&6	&6.45197 &6.49774	&[6.44401, 6.50102]\\
8		&4	&5.06733 &5.10055	&[5.06078, 5.10335]	&6	&5.9401 &5.98546	&[5.93238, 5.989]\\
10		&4	&4.85039 &4.88637	&[4.84386, 4.88953]	&6	&5.46672 &5.50782	&[5.45916, 5.51116]\\
12		&4	&4.63227 &4.66583	&[4.62575, 4.66884]	&6	&5.08473 &5.11779	&[5.07729, 5.12079]\\
14		&4	&4.43104 &4.45997	&[4.42453, 4.46279]	&6	&4.76734 &4.79957	&[4.76006, 4.80269]\\
16		&4	&4.25079 &4.27955	&[4.24441, 4.28237]	&6	&4.39537 &4.42311	&[4.38864, 4.42584]\\
18		&4	&4.07804 &4.10338	&[4.07164, 4.10605]	&6	&4.18139 &4.20744	&[4.17473, 4.21004]\\
20		&4	&3.94562 &3.96789	&[3.93923, 3.97034]	&6	&4.02465	&4.04716	&[4.01803, 4.04968]\\
\hline\hline
1		&8	&9.60253	&9.62348	&[9.59318, 9.62507]	&10	&11.0436	&11.0661	&[11.0332, 11.0677]\\
2		&8	&8.97188	&9.01806	&[8.96279, 9.02078]	&10	&10.0822	&10.1411	&[10.0721, 10.1443]\\
3		&8	&8.48335	&8.53629	&[8.47425, 8.53952]	&10	&9.31393	&9.3793	&[9.30393, 9.38283]\\
4		&8	&8.00789 &8.06203	&[7.99887, 8.06551]	&10	&8.58082	&8.63832	&[8.57102, 8.64178]\\
5		&8	&7.5251 &7.57926	&[7.5161, 7.58278]	&10	&7.9058	&7.961	&[7.89611, 7.96454]\\
6		&8	&7.06562 &7.11754 &[7.05669, 7.12102]	&10	&7.33533	&7.38481	&[7.32577, 7.38835]\\
8		&8	&6.18418 &6.23051 &[6.17596, 6.23403]	&10	&6.20274	&6.24781	&[6.19445, 6.2513]\\
10		&8	&5.54885 &5.58774 &[5.54107, 5.59089]	&10	&5.54885	&5.58774	&[5.54107, 5.59089]\\
\hline
\end{tabular}
\caption{\textit{Off-peak volume constraints.} Numerical results for the off-peak swing option for various exercise rights and refraction periods. The simulations are based on the approximation $\widehat{\underline{Y}}^{l}_j$ to the Snell envelope via the stopping rule \eqref{eq:tau}.}
\label{table:prag_upper_vol}
\end{table}
%%%%%%%%%%%%%%%%%%%%%%%%%%%%%%%%%%%
%%%%% end tabular %%%%%%%%%%%%%%%%%
%%%%%%%%%%%%%%%%%%%%%%%%%%%%%%%%%%%
\cbx{We run the above algorithm with $N_1=10000$, $N_2=300000$, $N_3=2000$, and $N_4=100$ sample paths. The numerical results for this off-peak
swing option are presented in Table  \ref{table:prag_upper_vol} for various choices of the number $L$ of exercise rights and the length $\delta$ of the refraction
period. Notice that the dual representations yet available in the literature do not cover the case of a nontrivial refraction period
($\delta \neq 1$) in combination with nontrivial volume constraints ($v\neq1$).
}
Due to the feature of allowing for exercising twice on weekends, the swing option prices are now higher than in the example \cbx{with unit volume constraint.
Moreover, additional rights can now become beneficial in situations in which they could not be exercised under the unit volume constraint
(e.g. the additional 8th right when the refraction period is $\delta=14$). }
\cbx{As for accuracy, }we \cbx{again }observe that the \cbx{relative }length of the $95\%$ confidence interval is less than $1\%$ \cbx{in all cases, which demonstrates
that the algorithm performs equally well in the presence of volume constraints. }

%%%%%%%%%%%%%%%%%%%%%%%%%%%%%%%%%%%%
%%%%%%%%%% tabular %%%%%%%%%%%%%%%%%
%%%%%%%%%%%%%%%%%%%%%%%%%%%%%%%%%%%%
%\begin{table}[htbp]
%\begin{tabular}{l || l || l l l}
%$\delta$	&$L$  &$Y^{up,L}_0$ 	&upper bound using \cite{Bender_vol} \\	
%\hline \hline
% 		 1   &1	&1.86474 (0.0005)		&1.8638 (0.0019)\\
%     1   &2	&3.40779 (0.0006)		&3.4078 (0.003)\\
%     1   &3	&4.73543 (0.0007)		&4.7368 (0.0038)\\
%     1   &4	&5.91312 (0.0008)		&5.9170 (0.0045)\\
%     1   &5	&6.97174 (0.0008)		&6.98 (0.0052)\\
%     1   &6	&7.93364 (0.0008)		&7.9470 (0.0058)\\
%     1   &7	&8.81384 (0.0008)		&8.8327 (0.0062)\\
%     1   &8	&9.62348 (0.0008)		&9.6493 (0.0069)\\
%     1   &9	&10.3723 (0.0008)		&10.4040 (0.0074)\\
%     1   &10	&11.0661 (0.0008)	&11.1035 (0.0079)\\
%\hline
%\end{tabular}
%\caption{\textit{Off peak volume constraints.} A comparison between the our upper bounds and the upper bounds obtained via the algorithm from \cite{Bender_vol} for the case of unit refraction period. The continuation function is obtained using $10000$ paths. The upper bounds are simulated with $2000$ outer and $100$ inner paths. The regression basis consists of the identity and the payoff function.}
%\label{table:bender_vol}
%%compare also with \cite[Section3.3]{Bender2010} for the case of refraction period $1$.
%\end{table}
%%%%%%%%%%%%%%%%%%%%%%%%%%%%%%%%%%%%
%%%%%% end tabular %%%%%%%%%%%%%%%%%
%%%%%%%%%%%%%%%%%%%%%%%%%%%%%%%%%%%%

%%%%%%%%%%%%%%%%%%%%%%%%%%%%%%%%%%%
%%%%%%%%% tabular %%%%%%%%%%%%%%%%%
%%%%%%%%%%%%%%%%%%%%%%%%%%%%%%%%%%%
\begin{table}[htbp]
\begin{tabular}{l || l || l l l}
$\delta$	&$L$  &$Y^{up,L}_0$ 	&upper bound using \cite{Bender_vol} \\	
\hline \hline
 		 1   &1	&1.86485 (0.0019)		&1.8638 (0.0019)\\
     1   &2	&3.40832 (0.003)		&3.4078 (0.003)\\
     1   &3	&4.73509 (0.0037)		&4.7368 (0.0038)\\
     1   &4	&5.90956 (0.0043)		&5.9170 (0.0045)\\
     1   &5	&6.96665 (0.00047)		&6.98 (0.0052)\\
     1   &6	&7.92669 (0.005)		&7.9470 (0.0058)\\
     1   &7	&8.80743 (0.0055)		&8.8327 (0.0062)\\
     1   &8	&9.61643 (0.0058)		&9.6493 (0.0069)\\
     1   &9	&10.3642 (0.0061)		&10.4040 (0.0074)\\
     1   &10	&11.0553 (0.00064)	&11.1035 (0.0079)\\
\hline
\end{tabular}
\caption{\textit{Off-peak volume constraints.} A comparison between our upper bounds and the upper bounds obtained
via the algorithm from \cite{Bender_vol} for the case of unit refraction period. Standard deviations are displayed in parentheses.}
\label{table:bender_vol}
%compare also with \cite[Section3.3]{Bender2010} for the case of refraction period $1$.
\end{table}
%%%%%%%%%%%%%%%%%%%%%%%%%%%%%%%%%%%
%%%%% end tabular %%%%%%%%%%%%%%%%%
%%%%%%%%%%%%%%%%%%%%%%%%%%%%%%%%%%%

In the case of unit refraction period $\delta=1$, \cbx{upper price bounds for the off-peak swing option can also
be computed by the dual representation of \cite{Bender_vol} for the marginal price of a multiple exercise option. This approach
generalizes the ideas of \cite{MH2004}: An upper biased estimate for the marginal price of having an additional $l$th right is computed in terms
of one martingale and $(l-1)$ stopping times. By summing up these upper bounds for the marginal prices, one finally ends with
an upper biased estimate for the option price. This approach is based on the fact that, roughly speaking, under
the assumption of a trivial refraction period ($\delta=1$) optimal exercise times for the problem with $(l-1)$ rights
are also optimal for the problem with $l$ rights, if one adds
one additional exercise time in a clever way. This is clearly not possible in general in the presence of a nontrivial refraction period.
So it seems that this alternative approach cannot be easily generalized to include refraction periods.
}

Table \ref{table:bender_vol} compares the upper bounds obtained using our method and the method from \cite{Bender_vol}
for the unit refraction case $\delta=1$. We mention that in Table \ref{table:bender_vol}, the variance reduction method
from Remark \ref{rem:var_red} is \cbx{not applied for both algorithms. As both methods are run with the same number of
sample paths and the nested maximum in our method can be efficiently calculated by the recursion
formula in Proposition \ref{prop:recursive_max}, the computational effort is roughly the same for both algorithms. }
We observe that, as the number of exercise rights increases, our method of directly tackling the Snell envelope produces upper bounds that become lower than the
algorithm tackling the marginal values from \cite{Bender_vol}. Whereas the differences for $L=1,\ldots,4$ are numerically not
 significant yet, they however become noticeable starting from \cbx{$L=5$ }and \cbx{are striking }for e.g. $L=10$.
 We also note that the \cbx{larger }$L$, the better our method performs \cbx{concerning }the variance of the upper bounds. At large, we
 conclude that if one is \cbx{mainly }interested in the price (and not the marginal price) of the swing option, our \cbx{new }method
\cbx{performs better than the algorithm from }\cite{Bender_vol}. \cbx{Moreover, it is applicable to a larger class of problems.}

\medskip

%%%%%%%%%%%%%%%%% %%%%%%
%%% BEGIN References %%%
%%%%%%%%%%%%%%%%%%%%%%%%
%\bibliography{references3}        % references.bib is the name of the database
%\bibliographystyle{abbrvnat}                        %alphadin{abbrvnat}%{plain}

\end{document}